\documentclass[12pt]{article}

\usepackage{url}
\usepackage{bbm}
\usepackage{mathtools}
\usepackage{amssymb}
\usepackage{amsthm}
\usepackage{empheq}
\usepackage{latexsym}
\usepackage{enumitem}
\usepackage{eurosym}
\usepackage{dsfont}
\usepackage{appendix}
\usepackage{color} 
\usepackage{hyperref}
\usepackage{frcursive}
\usepackage[utf8]{inputenc}
\usepackage[T1]{fontenc}
\usepackage{geometry}
\usepackage{multirow}
\usepackage{todonotes}
\usepackage{lmodern}
\usepackage{anyfontsize}
\usepackage{pgfplots}
\usepackage{stmaryrd}
\usepackage{natbib}

\setcitestyle{numbers,open={[},close={]}}

\pgfplotsset{compat=1.14}

\definecolor{red}{rgb}{0.7,0.15,0.15}
\definecolor{green}{rgb}{0,0.5,0}
\definecolor{blue}{rgb}{0,0,0.7}
\hypersetup{colorlinks, linkcolor={red},citecolor={green}, urlcolor={blue}}
			
\makeatletter \@addtoreset{equation}{section}

\newtheorem{theorem}{Theorem}[section]
\newtheorem{assumption}[theorem]{Assumption}

\newtheorem{lemma}[theorem]{Lemma}

\newtheorem{remark}[theorem]{Remark}

\def \P{\mathbb{P}}

\def \R{\mathbb{R}}

\def \Z{\mathbb{Z}}

\def\Dc{{\cal D}}

\def\Fc{{\cal F}}

\title{ How to design a derivatives market?\footnote{This work benefits from the financial support of the Chaires Analytics and Models for Regulation, Financial Risk and Finance and Sustainable Development. The authors gratefully acknowledge the financial support of the ERC Grant 679836 Staqamof. The authors would like to thank Angelique Bégrand, Luxi Chen and Laurent Fournier from Euronext, and Gilles Pagès.}}
\author{Bastien {\sc Baldacci}\footnote{\'Ecole Polytechnique, CMAP, 91128, Palaiseau Cedex, France,  bastien.baldacci@polytechnique.edu}\and 
Paul {\sc Jusselin}\footnote{\'Ecole Polytechnique, CMAP, 91128, Palaiseau Cedex, France, paul.jusselin@polytechnique.edu} \and 
Mathieu {\sc Rosenbaum}\footnote{\'Ecole Polytechnique, CMAP, 91128, Palaiseau Cedex, France, mathieu.rosenbaum@polytechnique.edu} }

\begin{document}

\maketitle
\begin{abstract}

We consider the problem of designing a derivatives exchange aiming at addressing clients needs in terms of listed options and providing suitable liquidity. We proceed into two steps. First we use a quantization method to select the options that should be displayed by the exchange. Then, using a principal-agent approach, we design a make take fees contract between the exchange and the market maker. The role of this contract is to provide incentives to the market maker so that he offers small spreads for the whole range of listed options, hence attracting transactions and meeting the commercial requirements of the exchange.\\

\end{abstract}

\section{Introduction}\label{Section Introduction}

Nowadays a typical role of an exchange is to give the possibility to investors to buy or sell financial products on electronic platforms, in sufficiently large quantity and at a reasonable price. Therefore exchanges have to set up their markets in a relevant way in order to achieve this goal. The issues related to market design cover a wide range of topics, from the microstructure of electronic trading platforms to the basic question of selecting the products that will be traded on the exchange.\\

Recently many papers have focused on the microstructural aspects of market design. For example the way of choosing an optimal tick size is addressed in \cite{dayri2015largetick}, where the authors study the relations between tick size, volatility and bid-ask bounce frequency. In \cite{budish2015high,jusselin2019optimal}, the relevance of continuous trading and its comparison with a frequent batch auction system is discussed, while market fragmentation is analyzed in \cite{sophie2018market}. Macroscopic features have also been investigated, see for example \cite{kalagnanam2004auctions}, where different market structures are classified with respect to several criteria such as matching mechanism, information feedback and bid structure.\\

Most of the research on market design focuses on stock markets. However, even if exchanges concentrate a large part of their activities on simple products such as stocks or futures, many also offer to their clients the possibility to trade more complex financial instruments such as derivatives. Actually there is very few academic literature on derivatives market design, mostly addressing the relationship between stock and option markets. For example in \cite{mayhew2004exchanges} the authors investigate the factors influencing the selection of stocks for option listing. However, they neither question the optimality of those factors, nor search for more relevant ones. The papers dealing with market design can in fact be separated into two groups: the ones that review and try to understand market practice and those proposing a theoretical framework in order to help exchanges improve their market design. Surprisingly, to our knowledge, there is no paper of the last kind dealing with derivatives market. In this article we propose a first contribution in that direction.\\

We take the realistic point of view of an exchange who wants to organize, or reorganize, its derivatives market. We consider that the market is made of vanilla European options only, that we view as independent of the underlying. By this, we mean that we deal with options that are used as hedging instruments and whose prices are essentially fixed by supply and demand. Finally we suppose that the exchange has access to data allowing for the estimation of the distribution of options market demand. For example, if the exchange already has a derivatives market it can use its own data, otherwise that of other exchanges. We focus on two issues: selecting the options that are going to be traded and attracting liquidity on those options.\\

The first issue faced by the exchange is the choice of the derivatives offered to the clients. Obviously it is impossible for the exchange to propose all maturities and strikes on its platform. This would be very hard to manage from a technical point of view and it would be impossible to guarantee liquidity on each option. As the maturities are quite standardized, the main challenge relies in strikes selection satisfying clients needs. Therefore, we consider that the exchange's problem is to select $n$ call options (or equivalently $n$ strikes), with fixed maturity, with the aim of maximizing the clients satisfaction. We define a measurement of this satisfaction and write the exchange objective under the form of a quantization problems. We refer to \cite{graf2007foundations,pages2004optimal} for an introduction to quantization. Such approach allows the exchange to select automatically a set of options based only on the data at its disposal.\\

The next goal of the exchange is to attract liquidity on its platform in order to increase the amount of executed orders. To do so, one way is to use a make take fees system: the exchange typically associates a fee rebate to executed limit orders, while charging a transaction fee for market orders. This enables it to subsidise liquidity provision and tax liquidity consumption. In \cite{el2018optimal} the authors design the optimal make take fees policy for a market with one market maker and a single undeying asset. This work has been extended in \cite{baldacci2019optimal} to the case of multiple market makers. The general principle of the approach in \cite{baldacci2019optimal,el2018optimal} is to consider that the exchange offers a contract to the market maker whose pay-off depends on the market order flow he generates. The problem of the exchange then boils down into designing the optimal contract in order to optimize the number of transactions.\\

However, in our setting the problem faced by the exchange is more complex to several extents. The main difference with the framework of \cite{baldacci2019optimal, el2018optimal} is that the exchange has to manage several assets simultaneously, namely the different options quoted on the platform. In order to focus on this issue we assume that there is only one market maker setting bid and ask quotes for all available options. Another challenge for a derivatives exchange is the possible absence of quotations for far from the money options (or quotations with a too wide spread). Such issue arises essentially for commercial reasons. Indeed, an exchange does not wish to display to its clients a product with scarse liquidity. It wants to make sure that there is sufficient available volume on the market for the whole range of listed options. Therefore, the design of an optimal make take fees policy for options market must aim at providing incentives to the market maker to lower the spreads, notably for far from the money options.\\

To do so, we are inspired by \cite{baldacci2019optimal,el2018optimal}, using a principal-agent framework. The exchange (the principal) has to design a contract towards the market maker (the agent) that maximizes a certain utility that depends on the behavior of the market maker. The main point is that the market maker's behavior, here the quoted spread on every available option, cannot be dictated by the exchange and depends on the contract. For example if the contract offers high incentives for every executed ask market order, then it is likely that the ask price quoted by the market maker will be close to the mid price. Formally, for a given contract, the market maker determines its behavior by solving a stochastic control problem. Then in order to find the optimal contract, the exchange maximizes its expected utility over the set of admissible contracts, knowing the market maker's response to each contract.\\

The paper is organized as follows. In Section \ref{Section Selection of the strikes} we explain how an exchange can select the options that will be traded on its platform using only market data. Then in Section \ref{Section Incentive policy of the exchange} we design the optimal contract that the exchange should offer to the market maker in order to maximize liquidity. Proofs and technical results are relegated to the Appendix.

\section{Market driven selection of the listed options}\label{Section Selection of the strikes}

In this section we build a method for the exchange to select the strikes that are going to be traded on its platform. This approach uses only data from trades volume reports and is based on a quantization algorithm. We illustrate this method by numerical experiments using data provided by Euronext.

\subsection{How to choose the strikes in order to match market demand?}
\label{subsec:provide}

We consider European call options with strikes expressed in percentage of the spot price (in moneyness) and that the exchange wishes to select $n$ strikes.\footnote{We do not address here the problem of choosing the number of strikes to propose. This point is left for further research.} Choosing relevant strikes, the exchange's objective is to maximize the satisfaction of the investors. So, we focus in this section on the market taking side of the trading flow. Section \ref{Section Incentive policy of the exchange} will be rather devoted to market makers.\\

We measure the regret of a market taker associated to the execution of a market order as a function of the difference between the strike he would have ideally bought (or sold) and the strike he actually bought (or sold). More precisely, for a given maturity, consider strikes $K_1< \dots < K_n$ that represent the options listed by the exchange. When a market taker wants to buy an option with strike $K$ he sends a market order on the option whose strike is the nearest from $K$. Hence he buys (or sells) the option with strike $K_i$ where $i$ is such that  
$$
K_i = \underset{1\leq j \leq n}{\arg\min}|K -K_j|.
$$
We consider that the regret associated to this market order is $\rho(|K- K_i|)$ where $\rho$ is an increasing function. Note that the regret of the market order can be written
$$
\underset{1\leq j \leq n}{\min}\rho(|K - K_j|).
$$
We finally assume that the strike $K$ is randomly chosen according to the distribution $\mathbb{P}^{mkt}$. This probability measure represents the law of market demand. Thus the higher the demand for a given strike the higher the probability that $K$ is close to this strike. The exchange can easily estimate the distribution $\mathbb{P}^{mkt}$ using data from its own options market or from other exchanges. The average regret of a market order is therefore written
\begin{equation}
\label{eq:obj_function}
\mathbb{E}^{mkt}[\underset{1\leq j \leq n}{\min}\rho(|K - K_j|)],
\end{equation}
where $\mathbb{E}^{mkt}$ denotes the expectation when $K\sim\mathbb{P}^{mkt}$. The problem of the exchange is then to find the $n$-uplet $(K_i)_{1\leq i \leq n}$ that minimizes \eqref{eq:obj_function}. Formally this corresponds to the following minimization problem:
\begin{equation}
\label{eq:quantization_problem}
\underset{K_1\leq \dots \leq K_n}{\arg\min} \mathbb{E}^{mkt}[\underset{1\leq j \leq n}{\min}\rho(|K - K_j|)].
\end{equation}
This type of optimization is classical in the field of signal or image processing and is called \textit{quantization} problem. The main idea of quantization is to summarize the information contained in a complex probability measure into a uniform probability with finite support. As an example, it allows to compress a signal (or an image) by selecting among its spectrum a given number of frequencies that summarizes the signal with the smallest possible loss of information. For an introduction to quantization problem see \cite{graf2007foundations, pages2004optimal}.\\

In this article we consider the quantization problem \eqref{eq:quantization_problem} when $\rho$ is a power-law function of the form $\rho(x) = |x|^{p}$ with $p\geq 2$. The power-law function has the advantage to be symmetric and convex. Therefore greater errors are increasingly penalized. As a consequence we expect the solution of \eqref{eq:quantization_problem} to capture the features of the tails of $\mathbb{P}^{mkt}$. Moreover the greater $p$, the more large errors are penalized. Hence for a large $p$, the $(K_i)_{1\leq i \leq n}$ solution of \eqref{eq:quantization_problem} are likely to be more spread towards large strikes and contain more extreme values of the distribution $\mathbb{P}^{mkt}$.

\subsection{Solving the quantization problem}
\label{subsec:solve_quantization}
 In this section we give some sufficient conditions that ensure that \eqref{eq:quantization_problem} has a unique solution. We also explain how \eqref{eq:quantization_problem} can be solved. \\

To get existence of a solution to the problem \eqref{eq:quantization_problem} we need to make the following assumption.
\begin{assumption}
\label{assumption:quantization} The probability $\mathbb{P}^{mkt}$ is absolutely continuous with respect to the Lebesgue measure with density that is log-concave and compactly supported in $[0,\overline{K}], \overline{K}>0$.
\end{assumption}
The assumption on the support of the probability is very reasonable since strikes between $0$ and $200\%$ of the spot price basically cover all the possible strikes of traded options. The log-concavity assumption is not really restrictive since it allows us to consider a wide class of probability distributions such as exponential type and Gaussian laws. It is shown in \cite[Theorem I-5.1]{graf2007foundations} that under Assumption \ref{assumption:quantization}, Problem \eqref{eq:quantization_problem} admits a unique non degenerate solution. The term non degenerate simply means that the optimal set of strikes satisfies $K_1<\dots < K_n$.\\

We now present a way to approximate numerically the solution of \eqref{eq:quantization_problem}. The idea behind the algorithm is that the solution $(K_i)_{1\leq i \leq n}$ can be seen as the fixed point of a function. This provides us a numerical method to approximate the $(K_i)_{1\leq i \leq n}$ that consists in iterating this function. This is known as the Lloyd's algorithm, which is a very intuitive approach that searches step by step the solution of \eqref{eq:quantization_problem}. A very convenient aspect of this algorithm is that it is automatic and easy to implement.\\

The Lloyd's algorithm starts with an initial set of strikes $(K_i)_{1\leq i \leq n}$ and is made of three steps:
 \begin{enumerate}
\item For any $i$, identify $A_i$ the set of "wished" strikes that corresponds to market orders sent to the strike $K_i$. Equivalently $A_i$ contains all the strikes $K$ which are closer to $K_i$ than from any other $K_j$
$$
A_i = \{K,\text{ s.t }i = \underset{1\leq j \leq n}{\arg\min}|K - K_j| \}.
$$
\item Set $K_i'$ as the unique strike in $A_i$ that minimizes the average regret of market orders sent with ideal strike in $A_i$. More precisely $K'_i$ is given by
$$
K'_i = \underset{k \in A_i}{\arg\min}~\mathbb{E}^{mkt}[|K-k|^p\mathbf{1}_{K\in A_i}].
$$
\item Go back to Step $1$ with $(K_i)_{1\leq i \leq n} = (K'_i)_{1\leq i \leq n}$ (or stop if a certain stopping criterion is reached and consider $(K'_i)_{1\leq i \leq n}$ as the approximate solution of \eqref{eq:quantization_problem}).
\end{enumerate}

The  Lloyd's algorithm has a very clear interpretation in terms of selecting the optimal set of strikes: first it identifies the area "controlled" by the $i-th$ strike and then improves the choice of the strikes. It is then intuitive that the solution of \eqref{eq:quantization_problem} is a fixed point of the Lloyd's algorithm. The sets $(A_i)_{1\leq i \leq n}$ form a covering of $\mathbb{R}_+$ that is often called the Vorono\"i tesselation associated to the $(K_i)_{1\leq i \leq n}$. It is easy to show that, for Step 1
$$
A_1 = [0, K_1],~~A_n = [K_n, \overline{K}]\text{ and for }i\in \{2, \dots, n-1\}:~~A_i = [\frac{K_i+K_{i-1}}{2}, \frac{K_{i+1}+K_{i}}{2}].
$$
A usual stopping criterion for Step $3$ is when $(K'_i)_{1\leq i \leq n}$ is too close from $(K_i)_{1\leq i \leq n}$. More precisely the algorithm stops if
$$
\sum_{i= 1}^n |K'_i - K_i|< \varepsilon,
$$
for a certain $\varepsilon >0$. Note that, starting from a discrete valued $\mathbb{P}^{mkt}$ (as will be the case here), when $p=2$, Step $2$ of the Lloyd's algorithm boils down to compute the average realization of $\mathbb{P}^{mkt}$ conditional on being in $A_i$. This can be obtained instantaneously. However when $p>2$, Step $2$ is not straightforward to compute in general. Yet the objective function being convex and taking the derivative with respect to $k$, a necessary and sufficient condition for $k$ to be solution of Step $2$ is
\begin{align*}
\mathbb{E}[|K-k|^{p-2}(K-k)\mathbf{1}_{K\in A_i}]=0    
\end{align*}
or equivalently
\begin{align*}
k = \frac{\mathbb{E}^{mkt}[K |K-k|^{p-2}\mathbf{1}_{K\in A_i}]}{\mathbb{E}^{mkt}[|K-k|^{p-2}\mathbf{1}_{K\in A_i}]}.    
\end{align*}
This characterizes the solution of Step $2$ as a fixed point. Thus one usually replaces Step $2$ by its iterative version: 
$$
K'_i = \frac{\mathbb{E}^{mkt}[K |K-K_i|^{p-2}\mathbf{1}_{K\in A_i}]}{\mathbb{E}^{mkt}[|K-K_i|^{p-2}\mathbf{1}_{K\in A_i}]}.
$$
From now on, we call Lloyd's algorithm the initial algorithm where we replace Step $2$ by its approximate version. We prove in Appendix \ref{assumption:quantization} that $(K_i)_{1\leq i \leq n}$ is solution of \eqref{eq:quantization_problem} if and only if it is a fixed point of the Lloyd's algorithm. The great strength of this method is that it is easy to implement, transparent, and completely automatic. Note also that if $\mathbb{P}^{mkt}$ has a discrete support, say $10$ strikes, then the Lloyd's algorithm will not necessarily select those strikes as solution of \eqref{eq:quantization_problem}.\\

We now turn to numerical experiments illustrating the efficiency of our method.

\subsection{Application}
\label{subsec:application}

In this section we apply our methodology to market data. First we describe the data and then present our numerical results.

\subsubsection{Description of the data}
\label{subsubsec:empirical_analysis}

We use data from Euronext, one of the main stock and option exchanges in Europe. The dataset contains for every trading day from the $3$-rd of December 2018 to the $24$-th of May 2019 and for every available options the total number of trades (buy and sell) during the day. Our dataset is only made of transactions that occurred on the Euronext platform. In particular we neither use OTC data nor data from another exchange. We choose for our example the most standard call options in terms of underlying on Euronext, namely options on the CAC 40 index. We report in Table \ref{tab:info_maturity} the number of call options traded each month for different ranges of maturity and in Table \ref{tab:info_strike} the number of call options traded each month for each strike.\\

In Figure \ref{fig:data_description}, we display the empirical distribution of traded option strikes (for all maturities) and the quantile plot of the maturity distribution in log-scale. The distribution of the strikes is unimodal, concentrated near the money and skewed towards in the money strikes. In Figure \ref{fig:data_description_b}, we provide the empirical distribution of traded options strikes for different ranges of maturity. We see that the distribution of the strikes depends on the maturity. In particular, the variance of the distribution is increasing with the maturity. The skewness towards in the money strikes is present for any maturity.

\begin{figure}[tbph!]
\begin{center}
\includegraphics[width=8cm,height=6cm]{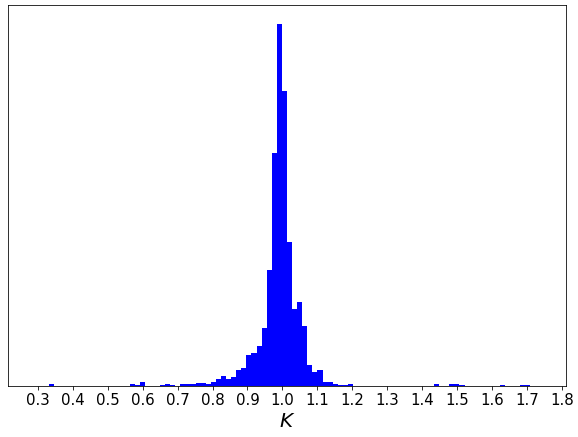}
\includegraphics[width=8cm,height=6cm]{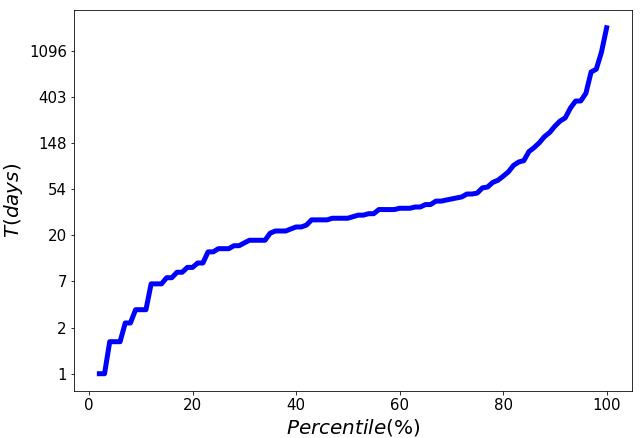}
\caption{Empirical distribution of traded option strikes (left). Quantile plot in log-scale of traded option maturities for the whole sample set (right).}
\label{fig:data_description}
\end{center}
\end{figure}

\begin{figure}[tbph!]
\begin{center}
\includegraphics[width=15cm,height=9cm]{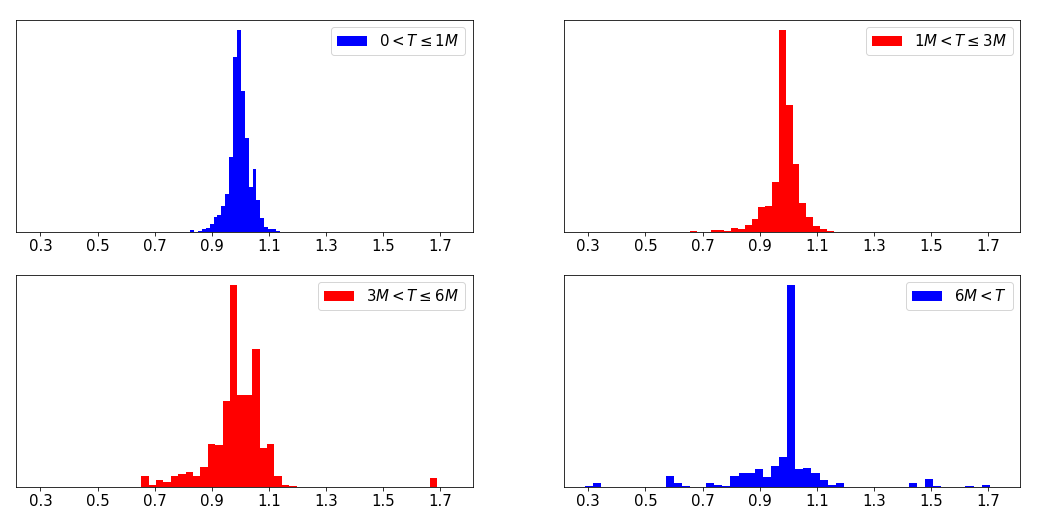}
\caption{Empirical distribution of the strikes for different maturities.}
\label{fig:data_description_b}
\end{center}
\end{figure}

\begin{table}[tbph!]
\centering
\begin{tabular}{c|cc|cc|cc}
 Maturity &  December&        January&       February&        March &      April &      May\\
 \hline
T$\leq$ 1M&  135951  &        99202  &       96323   &        191357&      161937&      108491\\
1M<T$\leq$ 3M&  79016   &        61651  &       30371   &        117400&      58914 &      121267\\
3M<T$\leq$ 6M&  10990   &        13279  &       15979   &        33901 &      11227 &      11779\\
6M < T&  71977   &        30278  &       14197   &        17158 &      25354 &      21330
\end{tabular}
\caption{Number of options traded by maturity and month.}
\label{tab:info_maturity}
\end{table}

\begin{table}[tbph!]
\centering
\begin{tabular}{c|cc|cc|cc}
  Strike (\%) &   December&        January&       February&        March &      April &      May\\
  \hline
20 &   0       &        0      &       0       &        0     &      55    &      10\\
30 &   1       &        1692   &       2       &        381   &      0     &      0\\
40 &   0       &        77     &       0       &        80    &      3     &      41\\
50 &   58      &        417    &       0       &        328   &      2031  &      1948\\
\hline
60 &   1933    &        152    &       31      &        323   &      691   &      2092\\
70 &   1402    &        1928   &       653     &        3837  &      2412  &      2956\\
80 &   12814   &        12952  &       3400    &        10118 &      14689 &      12147\\
90 &   113210  &        114463 &       10465   &        247877&      184835&      147362\\
\hline
100&   159075  &        68747  &       130002  &        94714 &      50621 &      90528\\
110&   5811    &        3586   &       12253   &        1766  &      83    &      2205\\
120&   869     &        94     &       64      &        11    &      0     &      16\\
130&   1       &        11     &       0       &        0     &      0     &      0\\
\hline
140&   0       &        0      &       0       &        0     &      2012  &      1960\\
150&   0       &        0      &       0       &        381   &      0     &      1602\\
160&   1720    &        271    &       0       &        0     &      0     &      0\\
170&   1040    &        20     &       0       &        0     &      0     &      0
\end{tabular}
\caption{Number of options traded by strike and month.}
\label{tab:info_strike}
\end{table}

\subsubsection{Numerical results}
\label{subsubsec:numerical_result_quantization}

We now present our numerical results. Since the distribution of the strikes depends on the maturity and because short maturities are over-represented in our data, we split our dataset into four subsets depending on the maturity: 
\begin{itemize}
\item maturity less than $1$ month,
\item maturity between $1$ and $3$ months,
\item maturity between $3$ and $6$ months,
\item maturity larger than $6$ months.
\end{itemize}

For any of those subsets we approximate the solution of the quantization problem \eqref{eq:quantization_problem} using the Lloyd's algorithm for $n=10$ and with stopping parameter $\varepsilon = 10^{-8}$. As initial value, we use $n$ points $(K_i)_{1\leq i\leq n}$ generated with uniform law between the $10$-th and $90$-th percentile of the dataset. In Figures \ref{fig:quantization_result_2} and \ref{fig:quantization_result_8} we plot a visualization of the quantization of the different sets obtained for $p=2$ and $p=8$.\\

\begin{figure}[]
\begin{center}
\includegraphics[width=15cm,height=9cm]{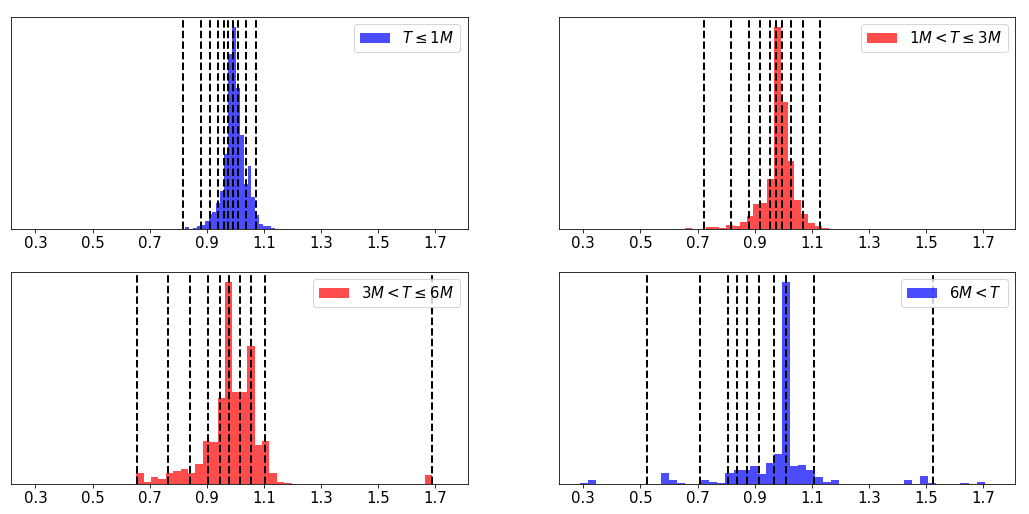}
\caption{Quantization of the option strikes using $p=2$ and $\varepsilon = 10^{-8}$. Empirical distribution of traded strikes is plotted in blue or red. The dotted lines correspond to the optimal quantization of $\mathbb{P}^{mkt}$.}
\label{fig:quantization_result_2}
\end{center}
\end{figure}

\begin{figure}[]
\begin{center}
\includegraphics[width=15cm,height=9cm]{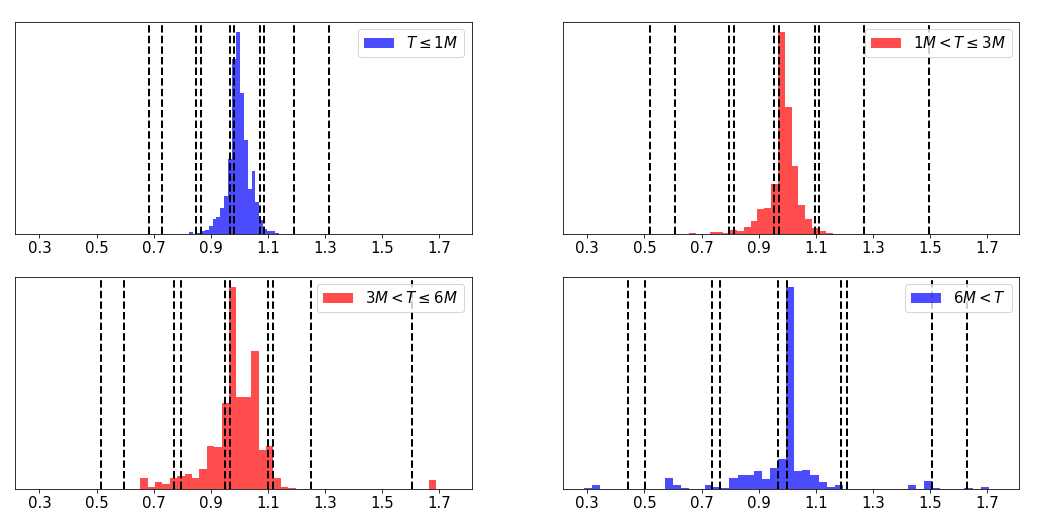}
\caption{Quantization of the option strikes using $p=8$ and $\varepsilon = 10^{-8}$. Empirical distribution of traded strikes is plotted in blue or red. The dotted lines correspond to the optimal quantization of $\mathbb{P}^{mkt}$.}
\label{fig:quantization_result_8}
\end{center}
\end{figure}

The strikes selected by the Lloyd's algorithm manage to reproduce some of the statistical properties of the demand distribution $\mathbb{P}^{mkt}$. In particular, for any range of maturity, the distribution of the $(K_i)_{1\leq i \leq n}$ is skewed towards in the money strikes. Also the variance of the selected strikes is increasing with the maturity as for market data.\\

We observe that for $p=8$ the strikes selected by the quantization method are more spread towards large strikes than for $p=2$. This is not surprising since the penalization of large errors is increasing with $p$ for the regret function $|\cdot|^p$. Therefore, as expected, the larger $p$, the more the solution of the quantization problem \eqref{eq:quantization_problem} contains extreme values of the distribution $\mathbb{P}^{mkt}$. We also note that the selected strikes for $p=8$ exhibit some kind of redundancy: some of them are very close to each other. In practice, one would of course discard one of two strikes being very close (it may then be interesting to take a smaller $n$). For practical applications, the easiest approach is probably to use $p=2$. With this choice, the Lloyd's algorithm is very fast and easy to implement. It also corresponds to the most documented case.\\

Finally we insist on the fact that when an exchange uses our methodology for strikes selection, it is interesting, if possible, to include transactions from other exchanges and from the OTC market in the dataset. This is because using only its own trade data may induce a bias in the strikes selection. For example if for some reasons clients of an exchange go on other venues to buy (or sell) out of the money options, then, in the exchange dataset, there will be very few transactions reported on out of the money options. This will lead to inaccuracies since the demand for out of the money options will be underestimated. \\

We now turn to the problem of providing incentives to the market maker to quote attractive spreads in order to attract liquidity towards the selected options.

\section{Incentive policy of the exchange}
\label{Section Incentive policy of the exchange}

In this section, we assume that the exchange has already selected a list of options. The goal is to design a contract between the exchange and the market maker so that the latter receives incentives to provide suitable liquidity on all the options. We first describe the market and assumptions. In particular, due to the short time horizon we are working on, we can assume a Bachelier model for the underlying asset and constant delta for the options. Then, we introduce a class of tractable admissible contracts proposed to the market maker. These contracts are indexed on the transactions induced by the behavior of the market maker. We show that there is no loss of generality in considering such class of contracts. For a given contract, the market maker solves an optimization problem to deduce its optimal quotes for each option. Then, the exchange maximizes his expected utility over the set of admissible contracts, knowing the response of the market maker to a given contract.\\

The utility of the exchange is made of two parts: one component related to the actual Profit and Loss (PnL for short) due to transactions, and one aiming at ensuring that enough liquidity is constantly posted on every option. As explained in the introduction, this second component addresses commercial constraints in order to make the exchange competitive. In particular, our model is flexible and can be designed so that the exchange has more interest in reducing the spreads for far from the money options, although not very traded, than for near the money options. We derive explicitly the optimal incentives that should be offered, up to the resolution of a two-dimensional linear PDE.\\

We conclude this section with numerical results showing the impact of the incentive policy on the spread of the listed options.

\subsection{The market}
\label{Subsection The Market}

This section is devoted to the description of the market model.\\

We consider a finite trading horizon time $T>0$ and a probability space $(\Omega,\mathcal{F},\mathbb{P}^{0})$ under which all stochastic processes are defined. Following Section \ref{Section Selection of the strikes}, we work on a market where European call options with strike $k\in \mathcal{K}:=\{K_1, \dots, K_n\}$ and maturity $\tau \in \mathcal{T}:= \{T_1, \dots, T_m\}$ can be traded. We focus on call options but our results can be extended to put options in a straightforward manner. The price of the underlying, observable by all market participants, has a dynamic given by
\begin{align}
\label{Definition efficient price}
\mathrm{d}S_{t}=\sigma \mathrm{d}W_{t},
\end{align}
where $\sigma>0$ is the volatility of the asset and $W$ is a one-dimensional Brownian motion. The choice of an arithmetic Brownian motion is motivated by the fact that we use a reasonably short time horizon $T$ (less than one day). On such scale, Bachelier and Black-Scholes type dynamics are quite indistinguishable.\\

Assuming zero interest rate, we write the price at time $t$ of the call option with maturity $\tau$ and strike $k$ as $C_{t}^{k,\tau}$. Its dynamic is given by
\begin{align}\label{Definition dynamics option}
\mathrm{d}C_{t}^{k,\tau}=\sigma \Delta^{k,\tau}_t \mathrm{d}W_{t},
\end{align}
where $\Delta^{k, \tau}_t:=\mathcal{N}(d_t)$ is the Bachelier delta of the call option $C^{k,\tau}$ at time $t$, $\mathcal{N}(\cdot)$ is the cumulative distribution function of the standard Gaussian law and $d_t:=\frac{S_t-k}{\sigma\sqrt{\tau}}$.\\

As we work over a short time horizon, the delta of the quoted options does not vary significantly. Hence, throughout the paper, we assume it to be constant.
\begin{assumption}\label{Assumption delta constant}
We consider that
$$
\Delta^{k, \tau}_t = \Delta^{k, \tau}.
$$
\end{assumption}
The delta of each option therefore becomes a model parameter. If the exchange observes a significant underlying price move, he can recalibrate the deltas, which will lead to a different pay-off of the contract for the market maker. \\

The market maker displays bid and ask quotes on the listed options. The market maker best bid price and best ask price at time $t$ on the option with maturity $\tau$ and strike $k$ are
\[
P_{t}^{k,\tau,b}:=C_{t}^{k,\tau}-\delta_{t}^{k,\tau,b}, \; P_{t}^{k,\tau,a}=C_{t}^{k,\tau}+\delta_{t}^{k,\tau,a}, \; t\in[0,T],
\]
where the superscript $b$ (resp. $a$) stands for bid (resp. ask). So we consider that the market maker controls the spreads $\delta^{k, \tau}:=(\delta^{k, \tau, a}, \delta^{k, \tau, b})$ on each option. The set of admissible controls for the market maker is therefore defined as
\begin{align}
\mathcal{A}:=\big\{(\delta_{t})_{t\in [0,T]}=(\delta_{t}^{k,\tau,i})_{t\in [0,T]}, k\in \mathcal{K}, \tau\in \mathcal{T}, i\in \{a,b\}, \text{ predictable and s.t } |\delta_{t}^{k,\tau,i}|\leq \delta_{\infty} \big\},
\label{Admissible controls market maker}
\end{align}
where $\delta_\infty>0$ is a constant, assumed to be large enough to satisfy technical conditions (see Appendix \ref{proof:principal_verification}). In practice it is of course not restrictive to assume that the spreads are bounded.\\

We now describe the dynamics of the market order flow. For every listed option, the arrival of ask (resp. bid) market orders is modeled by a point process $N^{k,\tau,a}$ (resp. $N^{k,\tau,b}$). We expect the intensity of buy (resp. sell) market order arrivals to be a decreasing function of both the spread quoted by the market maker $\delta^{k,\tau}$ and the transaction cost $f^{k,\tau}$ collected by the exchange. This has quite natural interpretation as a wider spread or higher fee decreases the number of transactions on the considered option. Moreover, we know from the literature (see \cite{dayri2015largetick}, \cite{madhavan1997security} and \cite{wyart2008relation}) that the average number of trades per unit of time for single assets is a decreasing function of the ratio between spread and volatility. Assuming same kind of behavior for the options, this leads to the following form of the intensity function:
\[
\lambda^{k,\tau}(\delta_{t}^{k,\tau,i}):=A \exp\Big(-\frac{C}{\sigma} (\delta_{t}^{k,\tau,i}+f^{k, \tau}) \Big),\; 
\]
where $A$ and $C$ are positive constants that can be calibrated using market data, and $f^{k, \tau}$ represents the fee fixed by the exchange for each market order. Furthermore, we assume that all market orders are of unit size. \\

The main difficulty in our framework is that the market maker is dealing with multiple derivatives. If the market maker strategy depends on its inventory on each option, then the problem lies in dimension $n$, which becomes intricate for large $n$. However, we will see that we can circumvent this issue since in our case we can aggregate the risk factors related to the inventories through the delta weighted cumulated inventory: 
\begin{align}\label{Inventory Process}
\mathcal{Q}_t := \sum_{(k,\tau)\in \mathcal{K}\times\mathcal{T}} \Delta^{k, \tau} Q_t^{k, \tau},
\end{align}
where $Q_{t}^{k,\tau}:=N_{t}^{k,\tau,b}-N_{t}^{k,\tau,a}$ is the number of options $C^{k, \tau}$ held by the market maker at time $t$. Each inventory is weighted by the corresponding~$\Delta$ (see Section \ref{Subsection optimal contract market maker} for details). Thus, the quantity $\mathcal{Q}$ represents the marked-to-market value of the market maker's portfolio. It therefore contains the market risk carried by the market maker. For example an out of the money option will account for a small part of the total risk, and conversely for in the money options. Finally we consider that the market maker has a critical absolute inventory $\overline{q} \in \mathbb{N}$. The intensity of the orders arrival is then
\begin{align*}
\lambda^{k,\tau,i}:=\lambda^{k,\tau}(\delta_{t}^{k,\tau,i})\mathbbm{1}_{\{\phi(i)\mathcal{Q}_{t^{-}}>-\overline{q}\}} \text{ with } \phi(i):=  \left\{
    \begin{array}{ll}
        1 \text{ if } i=a  \\
        -1 \text{ if } i=b. 
    \end{array}
\right.
\end{align*}

\begin{remark}
Note that there is a direct link between the spread quoted by the market maker and his inventory process. Indeed a lower spread $\delta^{k,\tau,b}$ (resp. $\delta^{k,\tau,a}$)  on the bid (resp. ask) side of the listed option $C^{k,\tau}$ increases the intensity of orders arrival  $\lambda^{k,\tau,b}$ (resp. $\lambda^{k,\tau,a}$). This leads to an increase (resp. decrease) of the inventory process $Q^{k,\tau}$. In other words, the market maker skews his quotes depending on the level of its aggregated inventory.
\end{remark}

 \subsection{Market maker's problem and contract representation}\label{Subsection optimal contract market maker}
 
In this section we exhibit the class of contracts used by the exchange. We also explain and solve the market maker's problem for any admissible contract.\\

The PnL of the market maker is defined as the sum of the cash earned from his executed orders and of the value of his inventory on each traded option. Thus, using that $\sum_{(k,\tau)\in\mathcal{K}\times\mathcal{T}}Q_t^{k,\tau}\Delta^{k,\tau}S_t=\mathcal{Q}_t S_t$, it writes
\begin{align}
PL_{t}^{\delta}:=\mathcal{W}_{t}^{\delta}+\mathcal{Q}_t S_t,
\end{align}
where 
\begin{align*}
\mathcal{W}_{t}^{\delta}:=\sum_{(k,\tau)\in \mathcal{K}\times\mathcal{T}}\int_{0}^{t}P_{u}^{k,\tau,a}\mathrm{d}N_{u}^{k,\tau,a}-\int_{0}^{t}P_{u}^{k,\tau,b}\mathrm{d}N_{u}^{k,\tau,b}    
\end{align*}
stands for his cash process at time $t\in[0,T]$. This expression shows the relevance of the variable $\mathcal{Q}$ for the market maker. It represents the volatility of the market maker's PnL with respect to the underlying price movements. Using \eqref{Definition dynamics option}, a direct integration by parts leads to the following form of the PnL process:
\begin{align}
\nonumber
 PL_{t}^{\delta}& :=\underset{i\in\{a,b\}}{\sum}\sum_{(k,\tau)\in \mathcal{K}\times\mathcal{T}}\int_{0}^{t}\delta_{u}^{k,\tau,i}\mathrm{d}N_{u}^{k,\tau,i}+\mathcal{Q}_u\mathrm{d}S_u.
\end{align}

Moreover, the exchange offers to the market maker a contract $\xi$, namely an $\mathcal{F}_{T}$-measurable random variable, which is added to his PnL at the end of the trading period. This contract aims at incentivizing the market maker to reduce the spread quoted for each option. More details will be given in Section \ref{Section Principal's problem}. The contract depends on all the transactions occuring between time $0$ and time $T$, as well as on the efficient price moves.  \\

Thus taking an exponential utility function, the market maker maximizes the following functional of his wealth: 
\begin{align}\label{Market Maker's Problem}
V_{\text{MM}}(\xi):=\underset{\delta\in \mathcal{A}}{\text{sup }}\mathbb{E}^{\delta}\bigg[-\text{exp}\Big(-\gamma\big(\xi+PL_{T}^{\delta}\big)\Big)\bigg],
\end{align}
where $\gamma>0$ denotes the market maker's risk aversion parameter and $\mathbb{E}^{\delta}$ the probability measure associated to a given control process $\delta\in \mathcal{A}$, see Appendix \ref{subsubsection probability measure} for details. For the well-posedness of Equation \eqref{Market Maker's Problem}, we need integrability conditions on the contract $\xi$, see Appendix \ref{subsection wellposedness} for details.\\

Finally we consider that the market maker accepts a contract $\xi$ only if its associated optimal expected utility $V_{\text{MM}}(\xi)$ is above some fixed threshold $R<0$. This threshold, called reservation utility of the agent, is the critical utility value under which the market maker has no interest in the contract. This quantity has to be taken into account carefully by the exchange before proposing a contract to the market makers.\\

We now introduce the class of contracts proposed to the market maker. Given $Y_0>0$, and predictable processes $Z:=(Z^{C^{k,\tau}},Z^{k,\tau,i})_{k\in\mathcal{K}\; \tau\in\mathcal{T}\; i\in\{a,b\}}\in\mathcal{Z}$ (see Appendix \ref{subsection wellposedness} for a definition of $\mathcal{Z}$), we introduce a special class of remuneration $\xi=Y_{T}^{Y_{0},Z}$ of the form
\begin{align}\label{Process Yt}
 Y_{T}^{Y_{0},Z}\!:=&Y_{0}\!+\!\!\int_{0}^{T} \!\bigg(\!\underset{i=a,b}{\sum}\!\sum_{(k,\tau)\in \mathcal{K}\times\mathcal{T}}\!\!\!\!\!Z_{r}^{k,\tau,i}\mathrm{d}N_{r}^{k,\tau,i}\!\!+\!\! Z_{r}^{C^{k,\tau}}\!\!\mathrm{d}C_{r}^{k,\tau}\!\!\bigg)\!+\!\Big(\frac{1}{2}\gamma\sigma^{2}\Big(\!\!\!\!\!\!\!\!\sum_{(k,\tau)\in \mathcal{K}\times\mathcal{T}}\!\!\!\!\!\!\! \Delta^{k,\tau}(Z_{r}^{C^{k,\tau}}\!\!\!\!\!+\!\!Q_{r}^{k,\tau})\Big)^{2}\!\!\!-\!\! H(Z_{r},Q_{r}\!)\!\Big)\mathrm{d}r,
\end{align}
where for $(z,q) \in \mathbb{R}^{2\times \#\mathcal{K}\times \#\mathcal{T} }\times\mathbb{R}$ with $z:=(z^{k,\tau})_{(k,\tau)\in \mathcal{K}\times\mathcal{T}}$, the function $H$, called Hamiltonian of the market maker, is defined by\footnote{This Hamiltonian term appears naturally when applying the dynamic programming principle for the market maker's problem.} 
$$
H(z,q):=\underset{\delta\in \mathbb{R}^{2\times \#\mathcal{K}\times \#\mathcal{T}}}{\text{sup }}h(\delta,z,q)
$$
with
\begin{align*}
&h(\delta,z,q):=\sum_{i=a,b}\sum_{(k,\tau)\in\mathcal{K}\times\mathcal{T}}\gamma^{-1}\bigg(1-\text{exp}\Big(-\gamma\big(z^{k,\tau,i}+\delta^{k,\tau,i}\big)\Big)\bigg)\lambda^{k,\tau}(\delta^{k,\tau,i})\mathbbm{1}_{\{\phi(i)\mathcal{Q}>-\overline{q}\}}.
\end{align*}

Actually, it turns out that it is enough to consider contracts of the form \eqref{Process Yt}. More precisely, we show that any admissible contract (in the sense of the integrability conditions specified in Appendix \ref{Condition wellposedness control problems}), is of this form. We have the following lemma proved in Appendix \ref{Proof equivalence admissible contracts}.

\begin{lemma}
\label{Lemma Contract Representation}
Any contract $\xi$ satisfying \eqref{Condition wellposedness control problems} has a unique representation $\xi=Y_{T}^{Y_{0},Z}$ for some $\big(Y_{0},Z\big)\in\mathbb{R}\times\mathcal{Z}$. 
\end{lemma}

Furthermore, the terms defining \eqref{Process Yt} have natural interpretation. 
\begin{itemize}
 \item The compensation $Y_0$ is calibrated by the exchange to ensure the reservation utility constraint with level $R$ of the market maker.\footnote{From Theorem \ref{Theorem Market Maker}, we see that taking $Y_0=-\log(-R)$ ensures the reservation utility of the market maker. }
 \item The term $Z^{C^{k,\tau}}$ is the compensation given to the market maker with respect to the volatility risk induced by the option $C^{k,\tau}$.
 \item Each time a trade is executed on the ask (resp. bid) side for the option $C^{k,\tau}$, the market maker is compensated by the term $Z^{k,\tau,a}$ (resp. $Z^{k,\tau,b}$).
 \item The term $\frac{1}{2}\gamma\sigma^{2}\Big(\underset{k=1}{\overset{\mathcal{K}}{\sum}}\underset{\tau=1}{\overset{\mathcal{T}}{\sum}} \Delta(Z^{C^{k,\tau}}+Q^{k,\tau})\Big)^{2}-H(Z,Q)$ is a continuous coupon given to the market maker. 
\end{itemize}
When the market maker remuneration is $Y^{Y_0, Z}$, its optimal response can be computed explicitly as a functional of $Z$. 
\begin{theorem}
\label{Theorem Market Maker}
 For $\xi = Y^{Y_0, Z}$, the market maker utility is 
\begin{align*}
V_{\textup{MM}}(Y_{T}^{Y_{0},Z})=-\exp(-\gamma Y_{0}), 
\end{align*}
associated to the optimal bid-ask policy $\hat{\delta}_{t}^{k,\tau,i}(\xi):=\Delta^{i}(Z_{t}^{k,\tau,i})$, where  
\begin{align}\label{Optimal Response Agent}
\Delta^{i}(Z_{t}^{k,\tau,i}) := (-\delta_{\infty})\vee \bigg(-Z_{t}^{k,\tau,i}+\frac{1}{\gamma}\textup{log}\Big(1+\frac{\sigma \gamma}{C}\Big)\bigg)\wedge \delta_{\infty} \text{ for } (k,\tau,i)\in\mathcal{K}\times\mathcal{T}\times\{a,b\}.
\end{align}
\end{theorem}

Theorem \ref{Theorem Market Maker} provides the optimal response of the market maker to any contract of the form \eqref{Process Yt}, see Appendix \ref{Proof theorem market maker} for the proof. Moreover from Equation \eqref{Optimal Response Agent}, we get that the exchange can anticipate the optimal behavior of the market maker. It is therefore easy for the platform to compute its own utility for a given contract.

\subsection{Solving the exchange's problem} \label{Section Principal's problem}

In this section we formalize the goal of the exchange and solve the problem of designing the optimal contract.

\subsubsection{Description of the exchange's problem}

We recall that the exchange has two objectives. The first one is to receive a high number of trades to collect the associated fees. The second is to have small spreads on its platform, in particular for far from the money options for which spreads are typically large. This is because the clients want to have sufficient liquidity on the whole list of options.\\

In order to quantify the first objective, we introduce a weighted version of the total number of trades:
$$
\mathcal{N}_t = \sum_{i = a, b}\sum_{(k, \tau)\in \mathcal{K}\times \mathcal{T}}c^{k, \tau}N^{k, \tau, i}_t,
$$
where for any $(k, \tau)\in\mathcal{K}\times\mathcal{T}$, $c^{k, \tau}\geq 0$ represents the value attributed to a trade on the option $C^{k, \tau}$ by the exchange.\footnote{One can for example take $c^{k,\tau}=f^{k,\tau}$. In this case, $\mathcal{N}_T$ represents the total amount of fees collected by the exchange.} Hence the more the exchange wants to attract liquidity on the option $C^{k, \tau}$, the higher $c^{k, \tau}$ has to be. If the considered option is very liquid (at the money options for example), the exchange may choose a rather small $c^{k,\tau}$.\\

To take into account the second objective, we consider the following quantity
\begin{align}
\label{Penalty Principal}
\mathcal{L}^{\delta}_{T}:=\underset{i=a,b}{\sum}\sum_{(k,\tau)\in \mathcal{K}\times\mathcal{T}}\int_{0}^{T}\omega\big(\delta^{k, \tau, i}_t-\delta^{k,\tau}_{\infty}\big)\mathrm{d}N_{t}^{k,\tau,i},
\end{align}
where $\omega\in [0, 1)$,\footnote{The choice of $\omega \in [0,1)$ is for technical reasons only.} and $\delta^{k,\tau}_{\infty}$ can be seen as a spread threshold the exchange would like to impose to the market maker. The more important the second objective for the exchange, the closer to one $\omega$ has to be chosen.\\

We thus consider that the exchange is looking for the contract $\xi$ that maximizes the following quantity:
\begin{align}\label{Objective function principal Xi}
\mathbb{E}^{\hat{\delta}(\xi)}\Bigg[-\text{exp}\bigg(-\eta\Big(\mathcal{N}_{T}-\mathcal{L}^{\delta(\xi)}_{T}-\xi\Big)\bigg)\Bigg],    
\end{align}
where $\eta>0$ is the risk aversion of the exchange and $\hat{\delta}(\xi)$ denotes the optimal response of the market maker given the contract $\xi$. \\

According to Lemma \ref{Lemma Contract Representation}, we know that it is enough for the exchange to consider contracts of the form $Y_T^{Y_0, Z}$ with $(Y_0,Z)\in \mathbb{R}\times\mathcal{Z}$. So, \eqref{Objective function principal Xi} becomes
\begin{equation}
\label{exchange objective}
\mathbb{E}^{\hat{\delta}(Y^{Y_0,Z})}\Bigg[-\text{exp}\bigg(-\eta\Big(\mathcal{N}_{T}-\mathcal{L}^{\delta(Y^{Y_0,Z})}_{T}-Y_T^{Y_0,Z}\Big)\bigg)\Bigg].
\end{equation}
Moreover for a contract of the form $Y^{Y_0, Z}$, from Theorem \ref{Theorem Market Maker}, the exchange knows the best response $\hat{\delta}(Y^{Y_0, Z})$ of the market maker. Indeed we recall that the optimal controls are given by
$$
\hat{\delta}^{k, \tau, i}(Y^{Y_0,Z}) = \Delta^{i}(Z_{t}^{k,\tau,i}).
$$
It implies that
$$
\mathcal{L}^{\hat{\delta}(Y^{Y_0,Z})}_{T}= \mathcal{L}^{Z}_{T}:= \underset{i=a,b}{\sum}\sum_{(k,\tau)\in \mathcal{K}\times\mathcal{T}}\int_{0}^{T}\omega\big(\Delta^{i}(Z_{t}^{k,\tau,i})-\delta^{k,\tau}_{\infty}\big)\mathrm{d}N_{t}^{k,\tau,i}.
$$
As in \cite{el2018optimal}, we notice that for a given contract $Y^{Y_0, Z}$, the market maker's optimal response does not depend on $Y_{0}$. The exchange objective function \eqref{exchange objective} being decreasing in $Y_{0}$, the maximization with respect to $Y_0$ is achieved at the level $\hat{Y}_{0} = -\log(-R)$.\footnote{Note that $-\exp(-\hat{Y}_0)=R$.} Finally, the exchange problem becomes
\begin{align}
\label{Reduced exchange problem}
&V_{0}^{E}:= \underset{Z \in \mathcal{Z}}{\text{sup}}~\mathbb{E}^{\Delta(Z)}\Bigg[-\text{exp}\bigg(-\eta\Big(\mathcal{N}_{T}^{Z}-\mathcal{L}_{T}^Z-Y_{T}^{\hat{Y}_{0},Z}\Big)\bigg)\Bigg].
\end{align}

\subsubsection{Stochastic control approach for the reduced exchange problem}

In this section we solve the reduced exchange problem \eqref{Reduced exchange problem}. We characterize the optimal contract components $Z^\star$ and explain how to compute them in practice. \\

To solve this stochastic control problem, we study the associated Hamilton-Jacobi-Bellman (HJB for short) equation. This approach characterizes an optimal $Z^{\star}$ solving \eqref{Reduced exchange problem} under the form of a feedback function. The following result is proved in Appendix \ref{proof:principal_verification}.

\begin{theorem}
\label{Main theorem}
The maximization problem \eqref{Reduced exchange problem} admits a solution $Z^{\star}$ given by
\begin{equation}
\label{exchange optimal control}
 Z^{\star k,\tau,i}(t,\mathcal{Q}_{t^{-}}):= \frac{1}{a-b}\log\bigg( \frac{bx_2U\big(t,\mathcal{Q}_{t^{-}}\big)}{ax_1^{k, \tau}U\big(t,\mathcal{Q}_{t^{-}}-\Delta^{k,\tau}\phi(i) \big)}\bigg) \text{ and } Z^{\star C^{k,\tau}}(t,Q_{t}^{k, \tau}):=-\frac{\gamma}{\gamma+\eta}Q_{t}^{k,\tau},
\end{equation}
for $(k,\tau,i)\in\mathcal{K}\times\mathcal{T}\times\{a,b\}$, where $a$, $b$, $(x_1^{k, \tau})_{k\in \mathcal{K}, \tau \in \mathcal{T}}$ and $x_2$ are constants defined in Appendix \ref{proof:principal_verification} and where $\tilde{U}:=(-U)^{-\frac{C}{\sigma\eta(1-\omega)}}$ is the unique solution of the following linear PDE on $[0, T]\times \mathbb{R}$:
\begin{align}\label{HJB PDE after linearization}
\left\{
    \begin{array}{ll}
        0=\partial_{t}\tilde{U}(t,\mathcal{Q}) - \tilde{U}(t,\mathcal{Q})   \frac{C\gamma \eta}{\gamma+\eta}\frac{ \sigma }{2(1-\omega)}   \mathcal{Q}^{2}+\underset{i=a,b}{\sum}\underset{(k, \tau)\in \mathcal{K}\times \mathcal{T}}{\sum}\hat{C}^{k,\tau}\tilde{U}\big(t,\mathcal{Q}-\Delta^{k,\tau}\phi(i)\big)\mathbf{1}_{\phi(i)\mathcal{Q}>-\overline{q}}, \\
        \tilde{U}(T,\mathcal{Q})=1,
    \end{array}
\right.
\end{align}
where $\hat{C}^{k,\tau}$ are defined in Apppendix \ref{proof:principal_verification}.
\end{theorem}

Theorem \ref{Main theorem} provides the incentives $Z^{\star}$ that maximize the exchange expected utility function, see Appendix \ref{proof:principal_verification} for the proof. The optimal contract is therefore given by
\begin{align}
\label{eq:optimal_contract}
\xi^{\star}=Y^{\hat{Y}_0, Z^*}&=\hat{Y}_{0}+\int_0^T \underset{(k, \tau) \in \mathcal{K}\times \mathcal{T}}{\sum} \bigg(\underset{i=a,b}{\sum} Z_{r}^{\star k,\tau,i}\mathrm{d}N_{r}^{k,\tau,i}+Z_{r}^{\star C^{k,\tau}}\mathrm{d}C_{r}^{k,\tau}\bigg) \\
\nonumber &+\int_0^T \Big(\frac{1}{2}\gamma\sigma^{2}\Big(\underset{(k, \tau) \in \mathcal{K}\times \mathcal{T}}{\sum} \Delta^{k,\tau}(Z_{r}^{\star C^{k,\tau}}+Q_{r}^{k,\tau})\Big)^{2}-H(Z^\star_{r},Q_{r})\Big)\mathrm{d}r.
\end{align}
We now provide some comments on the interpretation of the optimal incentives.
\begin{itemize}
\item The term $\int_{0}^{T}Z^{\star C^{k,\tau}}_{u}\mathrm{d}C_{u}^{k,\tau}$ in the optimal contract corresponds to part of the inventory risk process of the market maker $(Q_{t}^{k,\tau}C_{t}^{k,\tau})_{t\in [0,T]}$ that is supported by the exchange. As in \cite{el2018optimal}, the proportion of risk handled by the platform on each option is $\frac{\gamma}{\gamma+\eta}$. Hence, the more risk averse the exchange, the smallest this proportion.
\item An application of Ito's formula gives the following approximation:
\begin{align}
\log\Big(\frac{\tilde{U}(t,\mathcal{Q})}{\tilde{U}\big(t,\mathcal{Q}-\Delta^{k,\tau}\phi(i)\big)}\Big)\approx\phi(i) 2\frac{\sigma}{C}(T-t)\tilde{C}\Delta^{k,\tau}\mathcal{Q},
\end{align}
where $\tilde{C}:=\frac{C\gamma \eta}{\gamma+\eta}\frac{ \sigma }{2(1-\omega)}$. Thus, when  the  aggregated inventory  is highly positive, the exchange provides incentives to the market maker so that it attracts buy  market  orders  and  tries  to  dissuade  him  to  accept  more  sell  market  orders,  and conversely for a negative inventory.
\item Numerically, we show that the incentive $Z^{\star k,\tau,a}$ and $Z^{\star k,\tau,b}$ given by \eqref{exchange optimal control} are increasing functions of the value $c^{k,\tau}$ that the principal associates to the option $C^{k,\tau}$. Hence, he logically provides higher incentives to an option he is more interested in.
\item Although the principal manages a large number of listed options, we circumvent the curse of dimensionality by working with the aggregated inventory process. Note that the pay-off of the optimal contract depends only on $t$ and $\mathcal{Q}$. Thus it is very easy to compute for the exchange at the end of the trading day.
\end{itemize}

In practice to implement the above methodology, one needs to compute the function $\tilde{U}$ in order to design the optimal contract. A first way to do this is to use a classical finite difference scheme on the PDE \eqref{HJB PDE after linearization}. In Section \ref{Section numerical results PAgent} we use this technique for some numerical experiments on our method.\\

Moreover, as PDE \eqref{HJB PDE after linearization} is linear, we can also resort to a probabilistic representation to compute $\tilde{U}$ using a Monte-Carlo method. More precisely we have the following result which is a direct consequence of the Feynman-Kac formula.
\begin{lemma}
\label{lemma:representation}
We have the following representation:
\begin{align}
\tilde{U}(t,q):=\mathbb{E}\bigg[\textup{exp}\Big(\int_{t}^{T}-\tilde{C}\big(\mathcal{Q}_{s}^{t,q}\big)^{2}+\sum_{i=a,b}\sum_{(k, \tau)\in \mathcal{K}\times \mathcal{T}} \overline{\lambda}_{s}^{k,\tau,i}\mathrm{d}s\Big) \bigg],
\end{align}
where 
$$
\mathcal{Q}_{s}^{t,q}=q+\int_{t}^{s}\underset{(k, \tau)\in \mathcal{K}\times \mathcal{T}}{\sum} \Delta^{k, \tau } \mathrm{d}\big(\overline{N}_{u}^{k,\tau,b}-\overline{N}_{u}^{k,\tau,a}\big),
$$
where for any $(k, \tau)\in \mathcal{K}\times \mathcal{T}$ and $i=a$ or $b$, $\overline{N}^{k,\tau,i}$ is a point process with intensity $\overline{\lambda}_{s}^{k,\tau,i}:=\hat{C}^{k,\tau}\mathbf{1}_{\{ \phi(i)\mathcal{Q}^{t,q}_{s^{-}}>-\overline{q}\}}$, with $\hat{C}^{k,\tau}$ defined in Appendix \ref{proof:principal_verification}.
\end{lemma}

The proof is in the same vein as \cite[Proposition 4.1]{el2018optimal}. We now turn to numerical illustrations of our make take fees policy.

\subsection{Numerical results}\label{Section numerical results PAgent}

For numerical experiments, we consider three options which are characterized by their delta. We fix the following parameters: $A=1.5s^{-1}, \sigma=C=0.3s^{-1/2},f^{k, \tau}=[0.5,0.8,0.8]$ the vector of fees, and $\delta^{k, \tau}_\infty = [2,3,3]$ the set of quotation thresholds. The first option is at the money, the second one is in the money and the third is out of the money, hence the following set of deltas $[0.5,0.8,0.2]$. Moreover, we take $\eta=1, \gamma=0.01,T=100s,\overline{q}=40$.\\

We analyze the impact of the penalty $\omega$ and the weight associated to each options $c^{k,\tau}$ in the value function of the exchange.\\

In Figure \ref{figure spread inventory 0}, we display the average bid-ask spread at initial time on each option for $\omega=0$, and $c^{k,\tau}$ being equal either to $0$ or $0.1$. We see that a higher $c^{k,\tau}$ leads to a decrease of the spread for the option $C^{k,\tau}$. This result is in line with the form of the incentives in Theorem \ref{Main theorem}. Indeed $Z^{\star k,\tau,i}$ is an increasing function of $c^{k,\tau}$ and $\hat{\delta}^{k,\tau, i}$ is a decreasing function of $Z^{\star k,\tau, i}$. Thus, increasing the interest of the principal for the option $C^{k,\tau}$ leads to a decrease of the spread proposed by the market maker on this option. This shows that the exchange has a direct control on each option he is interested in.

\begin{figure}[h!]
\begin{center}
    \includegraphics[width=0.70\textwidth]{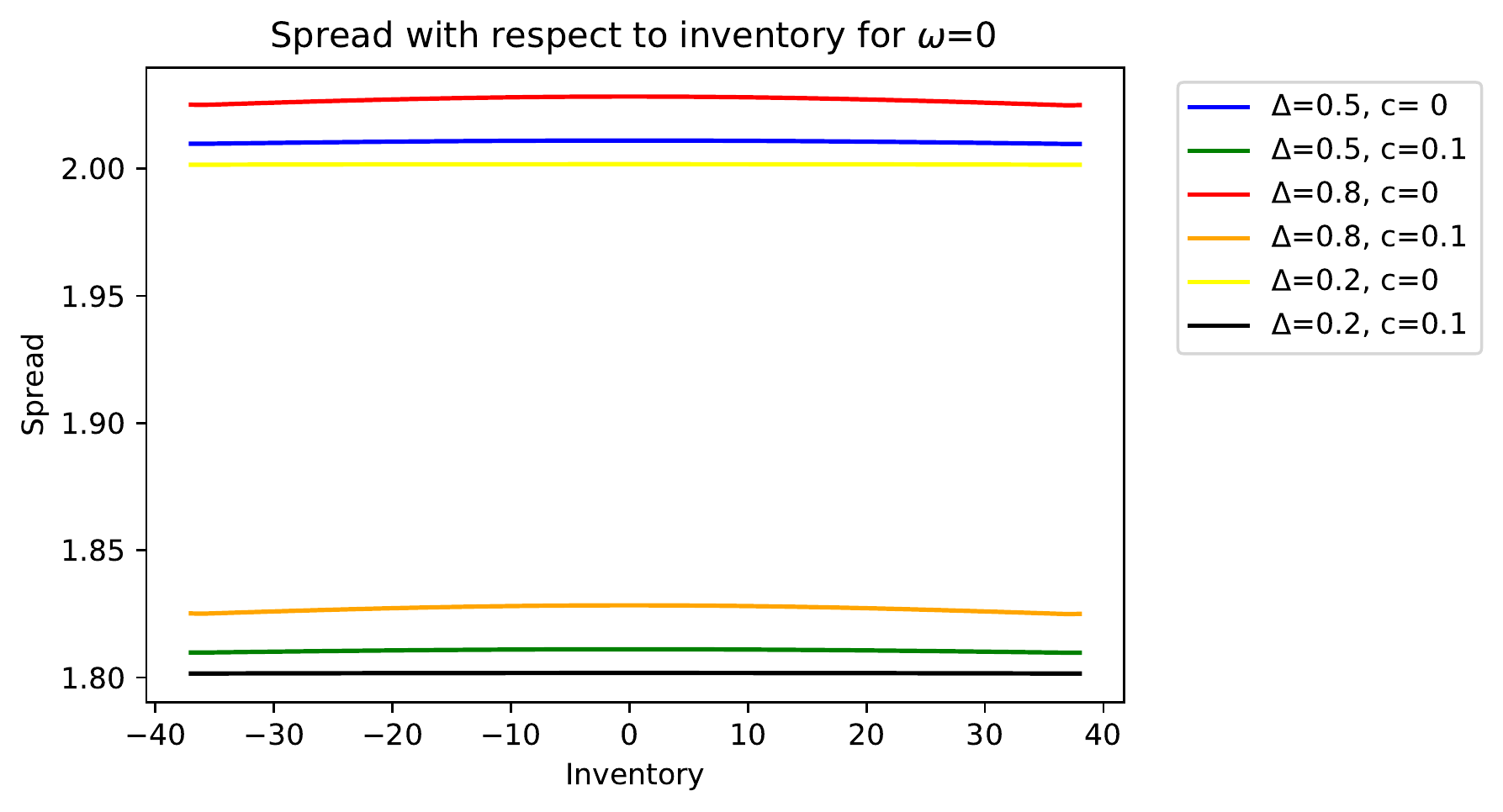}
\end{center}
\caption{Evolution of the spread at initial time with respect to inventory, $\omega=0$.}
\label{figure spread inventory 0}
\end{figure}

In Figure \ref{figure spread inventory 0.1}, we focus on the role of $\omega$, equal to $0.1$ on the spreads proposed by the market maker. As expected, a non-vanishing value of $\omega$ leads to a decrease of the spread for all the quoted options. This agrees with Theorem \ref{Main theorem}, where we see that the incentives are an increasing function of $\omega \in [0,1)$. Thus, the exchange can influence the whole set of spreads proposed on the quoted options. \\

\begin{figure}[h!]
\begin{center}
    \includegraphics[width=0.70\textwidth]{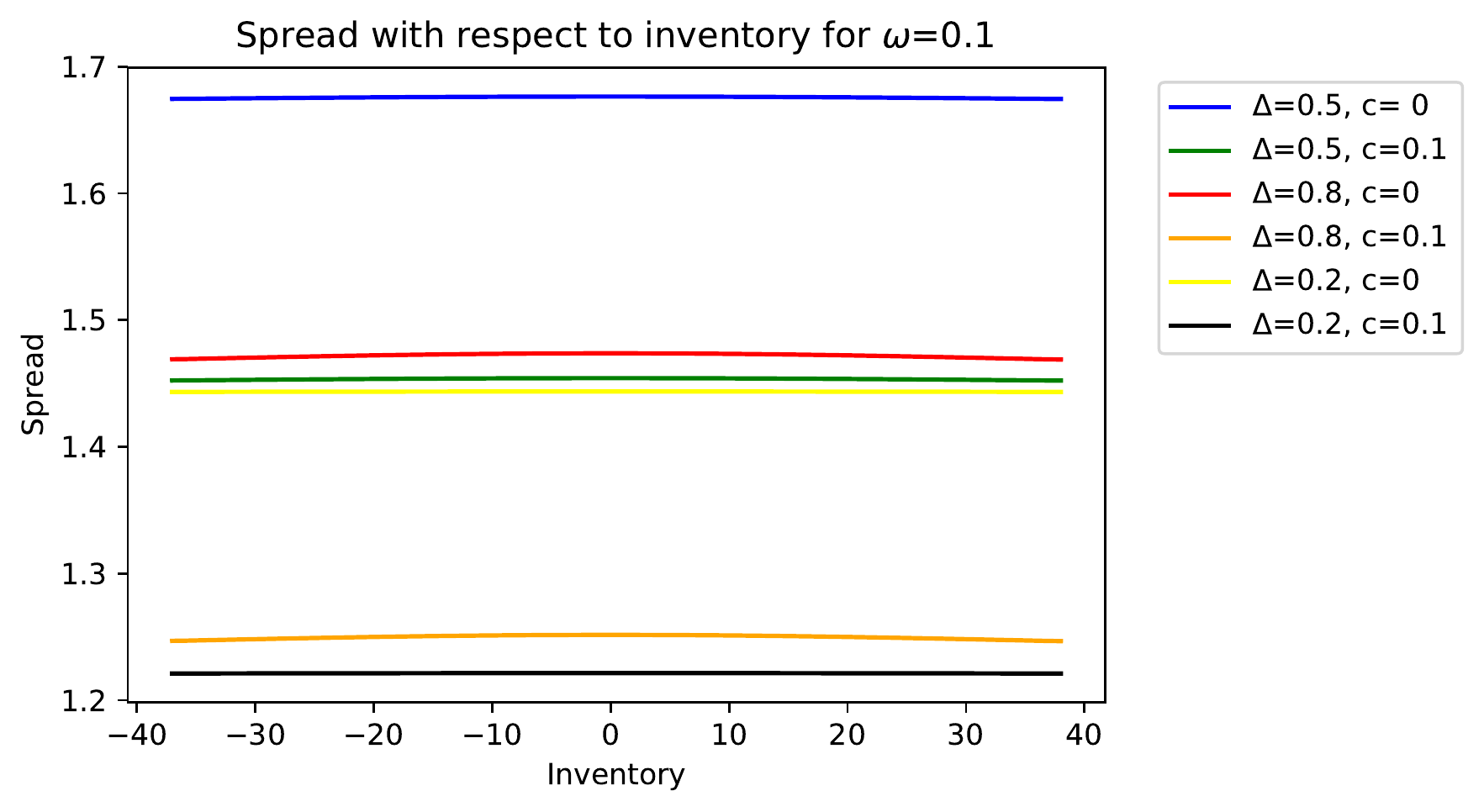}
\end{center}
\caption{Evolution of the spread at initial time with respect to inventory, $\omega=0.1$.}
\label{figure spread inventory 0.1}
\end{figure}

We conclude by showing in Figure \ref{figure spread inventory 0.2} the behavior of the average spread with a higher $\omega$, equal to $0.2$. We obtain similar effects as in Figure \ref{figure spread inventory 0.1}, namely a decrease of the spread on all quoted options for a higher $\omega$. 

\begin{figure}[h!]
\begin{center}
    \includegraphics[width=0.70\textwidth]{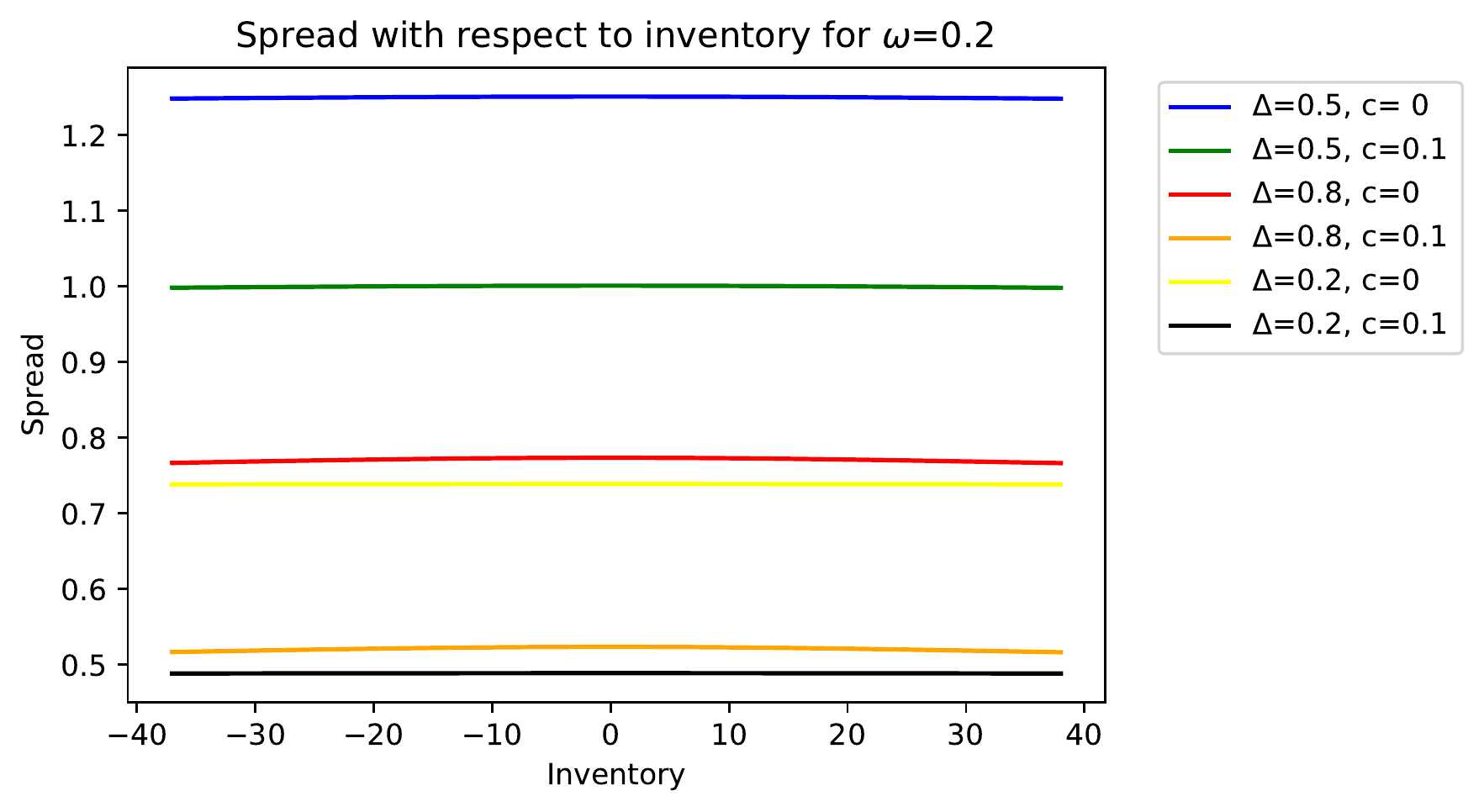}
\end{center}
\caption{Evolution of the spread at initial time with respect to inventory, $\omega=0.2$.}
\label{figure spread inventory 0.2}
\end{figure}

\subsection{Conclusion}

This work is, to our knowledge, the first to address the problem of designing a derivatives exchange, based solely on market data. In the first part, a simple market driven methodology enables us to choose which options the exchange should select to attract market takers. In the second part, we provide a make take fees policy between the exchange and the market maker which ensures a high quality of liquidity for the listed options.

\begingroup
\setcounter{section}{0}
\renewcommand\thesection{\Alph{section}}
\section{Appendix}
\subsection{Proof of the convergence of the Lloyd's algorithm}
\label{subsec:proof:proposition_quantization}

According to Paragraph 5.2 in \cite{graf2007foundations}, the set $(K_i)_{1\leq i \leq n}$ is a solution of \eqref{eq:quantization_problem} if and only if for any $i$, $A_i$ has positive Lebesgue measure and
$$
\int_{A_i}|K_i - x|^{p-1}\text{sgin}(x-K_i)\mathbb{P}^{mkt}(\mathrm{d}x) = 0,
$$
where $\text{sgin}$ is the sign function. This is equivalent to
$$
K_i = \frac{ \int_{A_i} |K_i - x|^{p-2} x\mathbb{P}^{mkt}(\mathrm{d}x) }{ \int_{A_i} |K_i - x|^{p-2}\mathbb{P}^{mkt}(\mathrm{d}x) } = \frac{ \mathbb{E}^{mkt}[|K_i - K|^{p-2} K\mathbf{1}_{K\in A_i}] }{ \mathbb{E}^{mkt} [|K_i - K|^{p-2} \mathbf{1}_{K\in A_i}] }.
$$
Thus $(K_i)_{1\leq i \leq N}$ is the solution of \eqref{eq:quantization_problem} if and only if it is a fixed point of the Lloyd's algorithm. \\

We now give proofs and technical results for Section \ref{Section Incentive policy of the exchange}. They are mostly inspired by \cite{el2018optimal}. However, for sake of completeness, we provide rigorous derivations. 

\subsection{Stochastic basis}
\subsubsection{Canonical process}
In this section, we give an accurate definition of the probability space defined in Section \ref{Subsection The Market}. We consider a final horizon time $T>0$ and the space $\Omega=:\Omega_{c} \times \Omega_{d}^{2\times\#\mathcal{T}\times\#\mathcal{K}}$, with $\Omega_{c}$ the set of continuous functions from $[0,T]$ into $\mathbb{R}$ and $\Omega_{d}$ the set of piecewise constant c\`adl\`ag functions from $[0,T]$ into $\mathbb{N}$. We consider $\Omega$ as a subspace of the Skorokhod space $\Dc([0,T],\R^{2\times\#\mathcal{T}\times\#\mathcal{K}+1})$ of c\`adl\`ag functions from $[0,T]$ into $\R^{2\times\#\mathcal{T}\times\#\mathcal{K}+1}$ and $\Fc$ the trace Borel $\sigma-$algebra on $\Omega$, where the topology is the one associated to the usual Skorokhod distance on $\Dc([0,T],\R^{2\times\#\mathcal{T}\times\#\mathcal{K}+1})$. \\

We define $(\mathcal{X}_{t})_{t\in [0,T]}:=\big(W_{t},(N_{t}^{k,\tau,i})_{i= a, b; k\in\mathcal{K}; \tau\in\mathcal{T} } \big)$ as the canonical process on $\Omega$, that is for any $\omega=(w,n^{k,\tau,i})\in \Omega$
\begin{align*}
W_{t}(\omega)=w(t), \; N_{t}^{k,\tau,i}(\omega)=n^{k,\tau,i}(t).
\end{align*}

\subsubsection{Probability measure}
\label{subsubsection probability measure}
We now properly define $\mathbb{P}^{0}$ and the associated change of measure. We set the probability $\mathbb{P}^{0}$ on $(\Omega,\Fc)$ such that under $\P^0$, $W$, $N^{k,\tau,i}$ are independent, $W$ is a one--dimensional Brownian motion and the $N^{k,\tau,i}, k\in\mathcal{K},\tau\in\mathcal{T},i=a,b$ are Poisson processes with intensity $\lambda^{k,\tau,i}(0)$.\footnote{In other words, $\P^0$ is simply the product measure of the Wiener measure on $\Omega_c$ and the unique measure on $\Omega_d^{2\times \#\mathcal{T}\times\#\mathcal{K}}$ that makes the canonical process an homogeneous Poisson process with the prescribed intensity.} Finally, we endow the space $(\Omega,\mathcal{F})$ with the ($\P^0-$completed) canonical filtration $\mathbb{F}:=(\mathcal{F}_{t})_{t\in[0,T]}$ generated by $(\mathcal{X}_{t})_{t\in[0,T]}$. \\

By \eqref{Admissible controls market maker}, the control process must be predictable and uniformly bounded. The last assumption is required to define the associated probability measure. So for $\delta \in \mathcal{A}$ we introduce the corresponding probability measure $\mathbb{P}^{\delta}$ under which $S_t= S_0+\sigma W_t$ follows \eqref{Definition efficient price} and for $k\in \mathcal{K}, \tau\in \mathcal{T}, i\in \{a,b\}$ the
$$
N_{t}^{\delta,k,\tau,i}:=N_{t}^{k,\tau,i}-\int_{0}^{t}\lambda^{k,\tau}(\delta_{r}^{k,\tau,i})\mathbbm{1}_{\{\phi(i)Q_{r^{-}}>-\overline{q}\}}\mathrm{d}r
$$
are martingales. This probability measure is defined by the corresponding Doléans-Dade exponential:

\begin{align*}
&L_{t}^{\delta}:=\text{exp}\Bigg(\underset{i=a,b}{\sum}\sum_{(k,\tau)\in \mathcal{K}\times\mathcal{T}}\int_{0}^{t}\mathbbm{1}_{\{\phi(i)Q_{r^{-}}>-\overline{q}\}}\bigg(\text{log}\Big(\frac{\lambda^{k,\tau}(\delta_{r}^{k,\tau,i})}{A}\Big)\mathrm{d}N_{r}^{k,\tau,i}-\Big(\lambda^{k,\tau}(\delta_{r}^{k,\tau,i})-A\Big)\mathrm{d}r\bigg)\Bigg),
\end{align*}
which is a true martingale by the uniform boundedness of $\delta_{t}^{k,\tau,i}$.\footnote{The associated Novikov criterion is given in \cite{sokol2013optimal}.} We can therefore define the Girsanov change of measure $\frac{\mathrm{d}\mathbb{P}^{\delta}}{\mathrm{d}\mathbb{P}^{0}}|_{\mathcal{F}_{t}}=L_{t}^{\delta}$, $\text{ for all } t\in [0,T]$. In particular, all the probability measures $\mathbb{P}^{\delta}$ indexed by $\delta \in \mathcal{A}$ are equivalent. We shall write $\mathbb{E}^{\delta}_{t}$ for the conditional expectation with respect to $\mathcal{F}_{t}$ under the probability measure $\mathbb{P}^{\delta}$. 

\subsection{Well-posedness of the optimization problems}\label{subsection wellposedness}
We give in this section the necessary integrability conditions ensuring that both exchange and market maker's problems are well defined. We consider the following assumptions: 
\begin{align}\label{Condition wellposedness control problems}
\underset{\delta \in \mathcal{A}}{\text{sup }}\mathbb{E}^{\delta}\bigg[\text{exp}\Big(-\gamma'\xi\Big)\bigg]<+\infty, \text{ for some } \gamma'>\gamma, \quad \underset{\delta \in \mathcal{A}}{\text{sup }}\mathbb{E}^{\delta}\bigg[\text{exp}\Big(\eta'\xi\Big)\bigg]<+\infty, \text{ for some } \eta'>\eta. 
\end{align}
Moreover, the next technical assumption is required in order to derive the best response of the market maker in Theorem \ref{Theorem Market Maker}:
\begin{align}
\underset{\delta \in \mathcal{A}}{\text{sup}}\text{ }\underset{t\in [0,T]}{\text{sup}}\mathbb{E}^{\delta}\bigg[\text{exp}\Big(-\gamma' Y_{t}^{0,Z}\Big)\bigg]< +\infty, \text{ for some } \gamma'>\gamma.
\label{Strong integrability MM}
\end{align}
Finally, we define $\mathcal{Z}$ as the set of predictable processes $(Z_t)_{t\in[0,T]}$ such that Conditions \eqref{Condition wellposedness control problems} and \eqref{Strong integrability MM} are satisfied. This is the set of admissible contract components of the exchange. 

\subsection{Dynamic programming principle}
In the spirit of \cite{el2018optimal}, we provide a proof of a dynamic programming principle for the market maker's problem. Note that a same type of dynamic programming principle exists for the exchange's problem. \\

For any $\mathbb{F}$ stopping time $\tau \in [t,T]$ and $\mu \in \mathcal{A}_{\tau}$, we define

\begin{align*}
&J_{T}(\tau,\mu)=\mathbb{E}_{\tau}^{\mu}\bigg[-\text{exp}\bigg(-\gamma\Big(\xi+\underset{i=a,b}{\sum}\sum_{(k,\tau)\in \mathcal{K}\times\mathcal{T}}\int_{\tau}^{T}\mu_{u}^{k,\tau,i}\mathrm{d}N_{u}^{k,\tau,i}+Q^{k,\tau}_{u}\mathrm{d}C_{u}^{k,\tau}\Big)\bigg)\bigg]
\end{align*}

where $\mathcal{A}_{\tau}$ denotes the restriction of $\mathcal{A}$ to controls on $[\tau,T]$. We also define the set $\mathcal{J}_{\tau,T}=(J_T(\tau,\mu))_{\mu \in \mathcal{A}_{\tau}}$. The continuation utility of the market maker is defined for any
$\mathcal{F}$-stopping time $\tau$ by 

\begin{align*}
V_{\tau}=\underset{\mu\in \mathcal{A}_{\tau}}{\text{ess sup }}J_{T}(\tau,\mu).
\end{align*}

We first prove the following technical lemma.

\begin{lemma}\label{Lemma 6.1}
Let $\tau$ be a stopping time with values in $[t,T]$. Then there exists an increasing sequence $(\mu^{n})_{n\in \mathbb{N}}$ in $\mathcal{A}_{\tau}$ such that $V_{\tau}=\lim_{n\rightarrow +\infty}J_{T}(\tau,\mu^{n})$.
\end{lemma}
\begin{proof}For $\mu, \mu' \in \mathcal{A}_{\tau}$ we define
\begin{align*}
\hat{\mu}:=\mu\mathbbm{1}_{\{J_{T}(\tau,\mu)\geq J_{T}(\tau,\mu') \} }+\mu'\mathbbm{1}_{\{J_{T}(\tau,\mu)\leq J_{T}(\tau,\mu') \} }.     
\end{align*}
We have $\hat{\mu}\in \mathcal{A}_{\tau}$ and by definition of $\hat{\mu}$, $J_{T}(\tau,\hat{\mu})\geq \max(J_{T}(\tau,\mu),J_{T}(\tau,\mu'))$. Thus $\mathcal{J}_{\tau,T}$ is increasing, and we obtain the same result as in \cite{el2018optimal}. The conclusion follows.
 \end{proof}

We set
\begin{align*}
\mathcal{D}_{t,T}(\delta):=\textup{exp}\bigg(-\gamma\Big(\underset{i=a,b}{\sum}\sum_{(k,\tau)\in \mathcal{K}\times\mathcal{T}}\int_{t}^{T}\delta_{u}^{k,\tau,i}\mathrm{d}N_{u}^{k,\tau,i}+Q^{k,\tau}_{u}\mathrm{d}C_{u}^{k,\tau}\Big)\bigg).
\end{align*}
Given Lemma \ref{Lemma 6.1}, we can now prove the dynamic programming principle associated to \eqref{Market Maker's Problem}.
\begin{lemma}\label{Lemma 6.2}
Let $t\in [0,T]$ and $\tau$ be an $\mathbb{F}$ stopping time with values in [t,T]. Then
\begin{align*}
&V_{t}=\underset{\delta\in \mathcal{A}}{\textup{ess sup }}\mathbb{E}_{t}^{\delta}\bigg[-\mathcal{D}_{t,\tau}(\delta)V_{\tau}\bigg]. 
\end{align*}
\end{lemma}
\begin{proof} Let $t\in [0,T]$ and $\tau$ be a stopping time with values in $[t,T]$. First, by tower property, we have
\begin{align*}
& V_{t}=\underset{\delta\in \mathcal{A}}{\text{ess sup }}\mathbb{E}_{t}^{\delta}\bigg[-\mathcal{D}_{t,T}(\delta)\text{exp}\big(-\gamma\xi\big)\bigg]\\
&\hspace{1em} =\underset{\delta\in \mathcal{A}}{\text{ess sup }}\mathbb{E}_{t}^{\delta}\bigg[\mathcal{D}_{t,\tau}(\delta)\mathbb{E}_{\tau}^{\delta}\Big[-\mathcal{D}_{\tau,T}(\delta)\text{exp}\big(-\gamma\xi\big)\Big]\bigg].
\end{align*}
Then, Bayes rule yields
\begin{align*}
&\mathbb{E}_{\tau}^{\delta}\bigg[-\mathcal{D}_{\tau,T}(\delta)\text{exp}\big(-\gamma\xi\big)\bigg]=\mathbb{E}^{0}_{\tau}\bigg[-\frac{L_{T}^{\delta}}{L_{\tau}^{\delta}}\mathcal{D}_{\tau,T}(\delta)\text{exp}\big(-\gamma\xi\big)\bigg]\\
& \hspace{11.5em} \leq \underset{\delta \in \mathcal{A}}{\text{ess sup }}\mathbb{E}_{\tau}^{\delta}\bigg[\mathcal{D}_{\tau,T}(\delta)\text{exp}\big(-\gamma\xi\big)\bigg]\\
&\hspace{11.5em}=V_{\tau}.
\end{align*}
Finally we obtain 
\begin{align*}
& V_{t}\leq \underset{\mu \in \mathcal{A}}{\text{ess sup }}\mathbb{E}^{\mu}_{t}\bigg[V_{\tau}\mathcal{D}_{t,\tau}(\delta)\bigg].
\end{align*}
We next prove the reverse inequality. Let $\delta\in \mathcal{A}$ and $\mu\in \mathcal{A}_{\tau}$. We define $(\delta\otimes_{\tau}\mu)_{u}= \delta_{u}1_{\{0\leq u \leq \tau\}} + \mu_{u}1_{\{\tau \leq u \leq T\}}$. Then $\delta\otimes_{\tau}\mu \in \mathcal{A}$ and by tower property
\begin{align*}
&V_{t}\geq \mathbb{E}_{t}^{\delta\otimes_{\tau}\mu}\bigg[-\mathcal{D}_{\tau,T}(\mu)\mathcal{D}_{t,\tau}(\delta)\text{exp}\big(-\gamma\xi\big)\bigg] =\mathbb{E}_{t}^{\delta\otimes_{\tau}\mu}\bigg[\mathbb{E}_{\tau}^{\delta\otimes_{\tau}\mu}\Big[-\mathcal{D}_{\tau,T}(\mu)\text{exp}(-\gamma\xi)\Big]\mathcal{D}_{t,\tau}(\delta)\bigg].
\end{align*}
Using Bayes formula and noting that $\frac{L_{T}^{\delta\otimes_{\tau}\mu}}{L_{\tau}^{\delta\otimes_{\tau}\mu}}=\frac{L_{T}^{\mu}}{L_{\tau}^{\mu}}$, we have
\begin{align*}
& \mathbb{E}_{\tau}^{\delta\otimes_{\tau}\mu}\Big[-\mathcal{D}_{\tau,T}(\mu)\text{exp}(-\gamma\xi)\Big]=\mathbb{E}^{0}_{\tau}\bigg[-\frac{L_{T}^{\mu}}{L_{\tau}^{\mu}}\mathcal{D}_{\tau,T}(\mu)\text{exp}(-\gamma\xi)\bigg]\\
& \hspace{12.5em} =J_{T}(\tau,\mu).
\nonumber
\end{align*}
This implies
\begin{align*}
& V_{t}\geq \mathbb{E}_{t}^{\delta\otimes_{\tau}\mu}\bigg[\mathcal{D}_{t,\tau}(\delta)J_{T}(\tau,\mu)\bigg].
\end{align*}
We can therefore use again Bayes rule and the fact that $\frac{L_{\tau}^{\delta\otimes_{\tau}\mu}}{L_{t}^{\delta\otimes_{\tau}\mu}}=\frac{L_{\tau}^{\delta}}{L_{t}^{\delta}}$ to obtain
\begin{align*}
& V_{t}\geq \mathbb{E}_{t}^{0}\bigg[\frac{L_{T}^{\delta\otimes_{\tau}\mu}}{L_{t}^{\delta\otimes_{\tau}\mu}}\mathcal{D}_{t,\tau}(\delta)J_{T}(\tau,\mu)\bigg]=\mathbb{E}_{t}^{0}\Bigg[\mathbb{E}_{\tau}^{0}\bigg[\frac{L_{T}^{\delta\otimes_{\tau}\mu}}{L_{\tau}^{\delta\otimes_{\tau}\mu}}\frac{L_{\tau}^{\delta\otimes_{\tau}\mu}}{L_{t}^{\delta\otimes_{\tau}\mu}}\mathcal{D}_{t,\tau}(\delta)J_{T}(\tau,\mu)\bigg]\Bigg] \\
&\hspace{14em} =\mathbb{E}_{t}^{0}\Bigg[\mathbb{E}_{\tau}^{0}\bigg[\frac{L_{T}^{\delta\otimes_{\tau}\mu}}{L_{\tau}^{\delta\otimes_{\tau}\mu}}\bigg]\frac{L_{\tau}^{\delta\otimes_{\tau}\mu}}{L_{t}^{\delta\otimes_{\tau}\mu}}\mathcal{D}_{t,\tau}(\delta)J_{T}(\tau,\mu)\Bigg]\\
&\hspace{14em} =\mathbb{E}_{t}^{0}\bigg[\frac{L_{\tau}^{\delta\otimes_{\tau}\mu}}{L_{t}^{\delta\otimes_{\tau}\mu}}\mathcal{D}_{t,\tau}(\delta)J_{T}(\tau,\mu)\bigg] \\
&\hspace{14em} =\mathbb{E}_{t}^{\delta}\bigg[\mathcal{D}_{t,\tau}(\delta)J_{T}(\tau,\mu)\bigg].
\end{align*}
Since the previous inequality holds for any $\mu\in \mathcal{A}_{\tau}$, we deduce from monotone convergence theorem together with Lemma \ref{Lemma 6.1} that there exists a sequence $(\mu^{n})_{n\in \mathbb{N}}$ of controls in $\mathcal{A}_{\tau}$ such that
\begin{align*}
& V_{t}\geq \underset{n\rightarrow +\infty}{\text{lim}}\mathbb{E}_{t}^{\delta}\bigg[\mathcal{D}_{t,\tau}(\delta)J_{T}(\tau,\mu^{n})\bigg]=\mathbb{E}_{t}^{\delta}\bigg[\mathcal{D}_{t,\tau}(\delta)\underset{n\rightarrow +\infty}{\text{lim}}J_{T}(\tau,\mu^{n})\bigg] \\
&\hspace{14em} =\mathbb{E}_{t}^{\delta}\bigg[\mathcal{D}_{t,\tau}(\delta)V_{\tau}\bigg].
\end{align*}
This concludes the proof. 
\end{proof}

\subsection{Proof of Lemma \ref{Lemma Contract Representation}}\label{Proof equivalence admissible contracts}

We divide the proof into six steps.\\

\textbf{Step 1: Derivation of the martingale representation}. \\

For $\delta \in \mathcal{A}$, it follows from the dynamic programming principle of Lemma \ref{Lemma 6.2} that the process 
\begin{align*}
& U_{t}^{\delta}=V_{t}\mathcal{D}_{0,t}(\delta)
\end{align*}
defines a $\mathbb{P}^{\delta}$-supermartingale for any $\delta\in \mathcal{A}$. By standard analysis, we may then consider it  in  its  càdlàg  version  (by  taking  right  limits  along  rationals). By  the  Doob-Meyer decomposition, we can write
$U_{t}^{\delta}=M_{t}^{\delta}-A_{t}^{\delta}$ where $M^{\delta}$ is a $\mathbb{P}^{\delta}$-martingale and $A_t^{\delta}=A_{t}^{\delta,c}+A_{t}^{\delta,d}$ is an integrable non-decreasing predictable process such that $A_{0}^{\delta,c}=A_{0}^{\delta,d}=0$ with pathwise continuous component $A^{\delta,c}$ and with $A^{\delta,d}$ a piecewise constant predictable process. \\

From the martingale representation theorem under $\mathbb{P}^{\delta}$, see Appendix A.1 in \cite{el2018optimal},  there exists $\tilde{Z}^{\delta}=(\tilde{Z}^{\delta,S},\tilde{Z}^{\delta,k,\tau,i})_{k\in\mathcal{K},\tau \in \mathcal{T}, i=a,b}$ predictable, such that
\begin{align*}
M_{t}^{\delta}=V_{0}+\int_{0}^{t}\tilde{Z}_{r}^{\delta,S}\mathrm{d}S_{r}+\underset{i=a,b}{\sum}\sum_{(k,\tau)\in \mathcal{K}\times\mathcal{T}}\int_{0}^{t}\tilde{Z}_{r}^{\delta,k,\tau,i}\mathrm{d}N_{r}^{\delta,k,\tau,i}.
\end{align*}
\textbf{Step 2: Boundedness of the value function}. \\

We show that $V$ is a negative process.  In fact, thanks to the uniform boundedness of $\delta\in \mathcal{A}$, we have that 
$$
\frac{L_{T}^{\delta}}{L_{t}^{\delta}}\geq \alpha_{t,T}= \text{exp}\bigg(-\underset{i=a,b}{\sum}\sum_{(k,\tau)\in \mathcal{K}\times\mathcal{T}}\frac{k}{\sigma}N_{T}^{k,\tau,i}-2\times\#\mathcal{T}\times\#\mathcal{K}A\mathrm{e}^{-\frac{kc_{\infty}}{\sigma}}(\mathrm{e}^{\frac{k}{\sigma}}+1)(T-t)\bigg),
$$
where $c_{\infty}:=\underset{k,\tau}{\text{max }}c^{k,\tau}$.
Therefore
\begin{align*}
V_{t}\leq \mathbb{E}^{0}_{t}\Big[-\alpha_{t,T}\text{exp}\Big(-\gamma(\delta_{\infty}\underset{i=a,b}{\sum}\sum_{(k,\tau)\in \mathcal{K}\times\mathcal{T}}N_{T}^{k,\tau,i}+\int_{t}^{T}Q^{k,\tau}_{u}\mathrm{d}C_{u}^{k,\tau})\Big)\mathrm{e}^{-\gamma \xi}\Big]<0.
\end{align*}

\textbf{Step 3: Identification of the coefficients (1/2)}. \\

Let $Y$ be  the  process  defined for any $t\in[0,T]$ by $V_{t}=-\mathrm{e}^{-\gamma Y_{t}}$. As $A^{\delta,d}$ is a predictable point process and the jumps of $N^{k,\tau,i}, i=a,b$ are totally inaccessible stopping times under $\mathbb{P}^{0}$, we have $\big\langle N^{k,\tau,i},A^{\delta,d}\big\rangle_t=0$ a.s. Using Ito’s formula, we obtain that
\begin{align*}
Y_{T}=\xi, \text{ and } \mathrm{d} Y_{t}=\underset{i=a,b}{\sum}\sum_{(k,\tau)\in \mathcal{K}\times\mathcal{T}}Z_{t}^{k,\tau,i}\mathrm{d}N_{t}^{k,\tau,i}+Z_{t}^{S}\mathrm{d}S_{t}-\mathrm{d}I_{t}-\mathrm{d}\tilde{A}_{t}^{d},
\end{align*}
with
\begin{align*}
& Z_{t}^{k,\tau,a}=-\frac{1}{\gamma}\text{log}\Big(1+\frac{\tilde{Z}_{t}^{\delta,k,\tau,a}}{U_{t^{-}}^{\delta}}\Big)-\delta_{t}^{k,\tau,a} \\
& Z_{t}^{k,\tau,b}=-\frac{1}{\gamma}\text{log}\Big(1+\frac{\tilde{Z}_{t}^{\delta,k,\tau,b}}{U_{t^{-}}^{\delta}}\Big)-\delta_{t}^{k,\tau,b} \\
& Z_{t}^{S}=-\frac{\tilde{Z}_{t}^{\delta,S}}{\gamma U_{t^{-}}^{\delta}}-\sum_{(k,\tau)\in \mathcal{K}\times\mathcal{T}}Q^{k,\tau}_{t^{-}}\Delta^{k,\tau} \\
& I_{t}=\int_{0}^{t}\Big(\overline{h}(\delta_{r},Z_{r},Q_{r})\mathrm{d}r -\frac{1}{\gamma U_{r}^{\delta}}\mathrm{d}A_{r}^{\delta,c}\Big) \\
& \overline{h}(\delta,Z_{t},Q_{t})=h(\delta,Z_{t},Q_{t})-\frac{1}{2}\gamma\sigma^{2}(Z_{t}^{S})^{2} \\
& \tilde{A}_{t}^{d}=\frac{1}{\gamma}\sum_{s\leq t}\text{log}\Big(1-\frac{\Delta A_{t}^{\delta,d}}{U_{t^{-}}^{\delta}}\Big).
\end{align*}
In particular, the last relation between $\tilde{A}^{d}$ and $A^{\delta,d}$ shows that $\Delta a_{t}\geq 0$ is independent of $\delta\in \mathcal{A}$, with $a_t=-\frac{A_{t}^{\delta,d}}{U_{t^{-}}^{\delta}}$ and abusing notations slightly, $\Delta a_{t}=-\frac{\Delta A_{t}^{\delta,d}}{U_{t^{-}}^{\delta}}$. \\

In order to complete the proof, we argue in the subsequent steps that $Z \in \mathcal{Z}$ and that, for $t\in [0,T]$, $A_{t}^{\delta,d}=-\sum_{s\leq t}U_{s^{-}}^{\delta}\Delta a_{s}=0$ so that $\tilde{A}_{t}^{d}=0$ and $I_{t}=\int_{0}^{t}\overline{H}(Z_{r},Q_{r})\mathrm{d}r$, where
\begin{align*}
\overline{H}(Z_{t},Q_{t})=H(Z_{t},Q_{t})-\frac{1}{2}\gamma\sigma^{2}(Z_{t}^{S})^{2}.    
\end{align*}

\textbf{Step 4: Identification of the coefficients (2/2).} \\

Since $V_{T}=-1$, we get that
\begin{align*}
& 0=\underset{\delta\in \mathcal{A}}{\text{sup }}\mathbb{E}^{\delta}[U_{T}^{\delta}]-V_{0}\\
& \hspace{0.4em}=\underset{\delta\in \mathcal{A}}{\text{sup }}\mathbb{E}^{\delta}[U_{T}^{\delta}-M_{T}^{\delta}]\\
& \hspace{0.4em}=\gamma\underset{\delta\in \mathcal{A}}{\text{sup }}\mathbb{E}^{0}\Big[L_{T}^{\delta}\int_{0}^{T}U_{r^{-}}^{\delta}(\mathrm{d}I_{r}-\overline{h}(\delta,Z_{r},Q_{r})\mathrm{d}r+\frac{\mathrm{d}a_{r}}{\gamma})\Big].
\end{align*}
Moreover, the controls being uniformly bounded, we have 
\begin{align*}
U_{t}^{\delta}\leq -\beta_{t}=V_{t}\text{exp}\bigg(-\gamma\Big(\delta_{\infty}\underset{i=a,b}{\sum}\sum_{(k,\tau)\in \mathcal{K}\times\mathcal{T}}N_{T}^{k,\tau,i}+\int_{0}^{t}Q^{k,\tau}_{u}\mathrm{d}C_{u}^{k,\tau}\Big)\bigg)<0.
\end{align*}
Then, using $A^{\delta,d}\geq 0, U^{\delta}\leq 0$ and $\mathrm{d}I_{t}-\overline{h}(\delta,Z_{t},Q_{t})\mathrm{d}t\geq 0$, we obtain 

\begin{align*}
& 0\leq \underset{\delta\in \mathcal{A}}{\text{sup }}\mathbb{E}^{0}\Big[\alpha_{0,T}\int_{0}^{T}-\beta_{r^{-}}\big(\mathrm{d}I_{r}-\overline{h}(\delta,Z_{r},Q_{r})\mathrm{d}r + \frac{\mathrm{d}a_{r}}{\gamma}\big)\Big]\\
& \hspace{0.4em} =-\mathbb{E}^{0}\Big[\alpha_{0,T}\int_{0}^{T}\beta_{r^{-}}\big(\mathrm{d}I_{r}-\overline{H}(Z_{r},Q_{r})\mathrm{d}r + \frac{\mathrm{d}a_{r}}{\gamma}\big)\Big].
\end{align*}

The quantities $\alpha_{0,T}\int_{0}^{T}\beta_{r^{-}}(\mathrm{d}I_{r}-\overline{H}(Z_{r},Q_{r}))\mathrm{d}r$ and $\alpha_{0,T}\int_{0}^{T}\beta_{r^{-}}\frac{\mathrm{d}a_{r}}{\gamma}$ being non-negative random variables, the result follows. \\

\textbf{Step 5: Admissibility of the process $Z$}. \\

As $\xi$ satisfies the conditions in \eqref{Condition wellposedness control problems}, to prove that $Z\in\mathcal{Z}$, it is enough to show that for some $p>0$
\begin{align*}
\underset{\delta\in \mathcal{A}}{\text{sup}}\underset{t\in [0,T]}{\text{sup }}\mathbb{E}^{\delta}[\exp(-\gamma(p+1)Y_{t})]<+\infty.
\end{align*}
Using Hölder inequality together with the boundedness of the intensities of the $N^{k,\tau,i}$, we have that $\underset{\delta\in \mathcal{A}}{\text{sup }}\mathbb{E}^{\delta}[|U_{T}^{\delta}|^{p'+1}]<+\infty$ for some $p'>0$. We deduce 

\begin{align*}
\underset{\delta\in \mathcal{A}}{\text{sup}}\underset{t\in [0,T]}{\text{sup }}\mathbb{E}^{\delta}[|U_{t}^{\delta}|^{p'+1}]=\underset{\delta\in \mathcal{A}}{\text{sup }}\mathbb{E}^{\delta}[|U_{T}^{\delta}|^{p'+1}]<+\infty
\end{align*}
because $U^{\delta}$ is a $\mathbb{P}^{\delta}$-negative supermartingale. The conclusion follows using again Hölder inequality, the uniform boundedness of the intensities of the $N^{k,\tau,i}$ and the fact that
\begin{align*}
& \exp(-\gamma Y_{t})=U^{\delta}_{t}\text{exp}\bigg(\gamma\Big(\underset{i=a,b}{\sum}\sum_{(k,\tau)\in \mathcal{K}\times\mathcal{T}}\int_{0}^{t}\delta_{u}^{k,\tau,i}\mathrm{d}N_{u}^{k,\tau,i}+Q^{k,\tau}_{u}\mathrm{d}C_{u}^{k,\tau}\Big)\bigg).
\end{align*}

\textbf{Step 6: Uniqueness of the representation}.\\

Let $(Y_{0},Z),(Y_{0}^{'},Z^{'})\in \mathbb{R}\times\mathcal{Z}$ be such that $\xi=Y_{T}^{Y_{0},Z}=Y_{T}^{Y_{0}^{'},Z^{'}}$. By following the lines of the verification argument in the proof of Theorem
\ref{Theorem Market Maker}, we obtain the equality $Y_{t}^{Y_{0},Z}=Y_{t}^{Y_{0}^{'},Z^{'}}$ using the fact that the value of the continuation utility of the market maker satisfies
$$
-\mathrm{e}^{-\gamma Y_{t}^{Y_{0},Z}}=-\mathrm{e}^{-\gamma Y_{t}^{Y_{0}^{'},Z^{'}}}= \underset{\delta\in \mathcal{A}}{\text{ess sup }}\mathbb{E}_{t}^{\delta}\Big[-\mathrm{e}^{-\gamma (PL_{T}^{\delta}-PL_{t}^{\delta}+\xi)}\Big].
$$
This in turn implies that for $t\in[0,T]$ $Z_{t}^{k,\tau,i}\mathrm{d}N_{t}^{k,\tau,i}=Z_{t}^{'k,\tau,i}\mathrm{d}N_{t}^{k,\tau,i}$ and $Z_{t}^{S}\sigma^{2}\mathrm{d}t=Z_{t}^{'S}\sigma^{2}\mathrm{d}t=\mathrm{d}\langle Y,S\rangle_{t}$. Consequently, $(Y_{0},Z)=(Y_{0}^{'},Z^{'})$.

\subsection{Proof of Theorem \ref{Theorem Market Maker}}\label{Proof theorem market maker}

Let $\xi=Y_{T}^{Y_{0},Z}$ with $(Y_{0},Z)\in \mathbb{R}\times\mathcal{Z}$. We  first  prove
that for an arbitrary set of controls $\delta \in \mathcal{A}$, we have $J_{\text{MM}}(\delta,\xi)\leq -\mathrm{e}^{-\gamma Y_{0}}$, where $J_{\text{MM}}(\delta,\xi)$ is such that $V_{\text{MM}}(\xi)=\underset{\delta\in \mathcal{A}}{\text{sup }}J_{\text{MM}}(\delta,\xi)$. Then we will see that this inequality is in fact an equality when the corresponding Hamiltonian $h(\delta,z,q)$ is maximized. Denote 
\begin{align*}
& \overline{Y}_{t}:=Y_{t}^{Y_{0},Z}+\underset{i=a,b}{\sum}\sum_{(k,\tau)\in \mathcal{K}\times\mathcal{T}}\int_{0}^{t}\delta_{u}^{k,\tau,i}\mathrm{d}N_{u}^{k,\tau,i}+Q^{k,\tau}_{u}\mathrm{d}C_{u}^{k,\tau}
\end{align*}
with $t\in [0,T]$. An application of Ito’s formula leads to
\begin{align*}
\mathrm{d}\mathrm{e}^{-\gamma\overline{Y}_{t}} =&\gamma\mathrm{e}^{-\gamma\overline{Y}_{t^{-}}}\Bigg(\!-\!(\sum_{(k,\tau)\in \mathcal{K}\times\mathcal{T}}\!\!\!Q_{t}^{k,\tau}\Delta^{k,\tau}+Z_{t}^{S})\mathrm{d}S_{t}+(H(Z_{t},Q_{t})-h(\delta,Z_{t},Q_{t}))\mathrm{d}t \\
& -\underset{i=a,b}{\sum}\sum_{(k,\tau)\in \mathcal{K}\times\mathcal{T}}\gamma^{-1}\bigg(1-\text{exp}\Big(-\gamma\big(Z_{t}^{k,\tau,i}+\delta_{t}^{k,\tau,i}\big)\Big)\bigg)\mathrm{d}N_{t}^{\delta,k,\tau,i}\Bigg).
\end{align*}
Thus $\mathrm{e}^{-\gamma\overline{Y}_{.}}$ is a $\mathbb{P}^{\delta}$-local submartingale.  Thanks to Condition \eqref{Strong integrability MM}, the uniform boundedness of the intensities of the $N^{k,\tau,i}$ and Hölder inequality, $\Big(\mathrm{e}^{-\gamma\overline{Y}_{t}}\Big)_{t\in [0,T]}$ is uniformly integrable and hence a true submartingale. Doob-Meyer decomposition theorem gives us that
$$
\int_{0}^{\cdot}\gamma\mathrm{e}^{-\gamma\overline{Y}_{t^{-}}}\Bigg(\!-(\!\!\sum_{(k,\tau)\in \mathcal{K}\times\mathcal{T}}\!\!\! Q_{t}^{k,\tau}\Delta^{k,\tau}+Z_{t}^{S})\mathrm{d}S_{t} -\underset{i=a,b}{\sum}\sum_{(k,\tau)\in \mathcal{K}\times\mathcal{T}}\!\!\!\gamma^{-1}\bigg(1-\text{exp}\Big(-\gamma\big(Z_{t}^{k,\tau,i}+\delta_{t}^{k,\tau,i}\big)\Big)\bigg)\mathrm{d}N_{t}^{\delta,k,\tau,i}\Bigg)
$$
is a true martingale. This implies that
\begin{align*}
&J_{\text{MM}}(\delta,\xi)=\mathbb{E}^{\delta}\Big[-\mathrm{e}^{-\gamma\overline{Y}_{T}}\Big]\\
&\hspace{5em}=-\mathrm{e}^{-\gamma Y_{0}}-\mathbb{E}^{\delta}\bigg[\int_{0}^{T}\gamma\mathrm{e}^{-\gamma\overline{Y}_{t^{-}}}\big(H(Z_{t},Q_{t})-h(\delta,Z_{t},Q_{t})\big)\mathrm{d}t\bigg] \\
&\hspace{5em}\leq -\mathrm{e}^{-\gamma Y_{0}}.
\end{align*}
In addition to this, the previous inequality becomes an equality if and only if $\delta$ is chosen as the maximizer of the Hamiltonian $h$, thus leading to the optimal quotes provided in Theorem \ref{Theorem Market Maker}. So we deduce $J_{\text{MM}}(\delta,\xi)=-\mathrm{e}^{-\gamma Y_{0}}$. Finally we have $V_{\text{MM}}(\xi)=-\mathrm{e}^{-\gamma Y_{0}}$ with optimal response $(\hat{\delta}_{t})_{t\in [0,T]}$.

\subsection{Proof of Theorem \ref{Theorem Market Maker}}\label{Proof theorem market maker}

Let $\xi=Y_{T}^{Y_{0},Z}$ with $(Y_{0},Z)\in \mathbb{R}\times\mathcal{Z}$. We  first  prove
that for an arbitrary set of controls $\delta \in \mathcal{A}$, we have $J_{\text{MM}}(\delta,\xi)\leq -\mathrm{e}^{-\gamma Y_{0}}$, where $J_{\text{MM}}(\delta,\xi)$ is such that $V_{\text{MM}}(\xi)=\underset{\delta\in \mathcal{A}}{\text{sup }}J_{\text{MM}}(\delta,\xi)$. Then, we will see that this inequality is in fact an equality when the corresponding Hamiltonian $h(\delta,z,q)$ is maximized. Denote 
\begin{align*}
& \overline{Y}_{t}:=Y_{t}^{Y_{0},Z}+\underset{i=a,b}{\sum}\sum_{(k,\tau)\in \mathcal{K}\times\mathcal{T}}\int_{0}^{t}\delta_{u}^{k,\tau,i}\mathrm{d}N_{u}^{k,\tau,i}+Q^{k,\tau}_{u}\mathrm{d}C_{u}^{k,\tau}
\end{align*}
with $t\in [0,T]$. A direct application of Ito’s formula leads to
\begin{align*}
\mathrm{d}\mathrm{e}^{-\gamma\overline{Y}_{t}} =&\gamma\mathrm{e}^{-\gamma\overline{Y}_{t^{-}}}\Bigg(\!-\!(\sum_{(k,\tau)\in \mathcal{K}\times\mathcal{T}}\!\!\!Q_{t}^{k,\tau}\Delta^{k,\tau}+Z_{t}^{S})\mathrm{d}S_{t}+(H(Z_{t},Q_{t})-h(\delta,Z_{t},Q_{t}))\mathrm{d}t \\
& -\underset{i=a,b}{\sum}\sum_{(k,\tau)\in \mathcal{K}\times\mathcal{T}}\gamma^{-1}\bigg(1-\text{exp}\Big(-\gamma\big(Z_{t}^{k,\tau,i}+\delta_{t}^{k,\tau,i}\big)\Big)\bigg)\mathrm{d}N_{t}^{\delta,k,\tau,i}\Bigg).
\end{align*}
Thus, $\mathrm{e}^{-\gamma\overline{Y}_{.}}$ is a $\mathbb{P}^{\delta}$-local submartingale.  Thanks to Condition \eqref{Strong integrability MM}, the uniform boundedness of the intensities of the $N^{k,\tau,i}$ and Hölder inequality, $\Big(\mathrm{e}^{-\gamma\overline{Y}_{t}}\Big)_{t\in [0,T]}$ is uniformly integrable and hence is a true submartingale. Doob-Meyer decomposition theorem gives us that

\begin{align*}
&\int_{0}^{\cdot}\gamma\mathrm{e}^{-\gamma\overline{Y}_{t^{-}}}\Bigg(\!-(\!\!\sum_{(k,\tau)\in \mathcal{K}\times\mathcal{T}}\!\!\! Q_{t}^{k,\tau}\Delta^{k,\tau}+Z_{t}^{S})\mathrm{d}S_{t} -\underset{i=a,b}{\sum}\sum_{(k,\tau)\in \mathcal{K}\times\mathcal{T}}\!\!\!\gamma^{-1}\bigg(1-\text{exp}\Big(-\gamma\big(Z_{t}^{k,\tau,i}+\delta_{t}^{k,\tau,i}\big)\Big)\bigg)\mathrm{d}N_{t}^{\delta,k,\tau,i}\Bigg)
\end{align*}
is a true martingale. This implies that
\begin{align*}
&J_{\text{MM}}(\delta,\xi)=\mathbb{E}^{\delta}\Big[-\mathrm{e}^{-\gamma\overline{Y}_{T}}\Big]\\
&\hspace{5em}=-\mathrm{e}^{-\gamma Y_{0}}-\mathbb{E}^{\delta}\bigg[\int_{0}^{T}\gamma\mathrm{e}^{-\gamma\overline{Y}_{t^{-}}}\big(H(Z_{t},Q_{t})-h(\delta,Z_{t},Q_{t})\big)\mathrm{d}t\bigg] \\
&\hspace{5em}\leq -\mathrm{e}^{-\gamma Y_{0}}.
\end{align*}
In addition to this, the previous inequality becomes an equality if and only if $\delta$ is chosen as the maximizer of the Hamiltonian $h$ thus leading to the optimal quotes provided in Theorem \ref{Theorem Market Maker}. So we deduce $J_{\text{MM}}(\delta,\xi)=-\mathrm{e}^{-\gamma Y_{0}}$. Finally we have $V_{\text{MM}}(\xi)=-\mathrm{e}^{-\gamma Y_{0}}$ with optimal response $(\hat{\delta}_{t})_{t\in [0,T]}$.

\subsection{Proof of Theorem \ref{Main theorem}}\label{proof:principal_verification}
We define for any map $v:[0,T]\times\mathbb{Z}^{\#\mathcal{K}\times\#\mathcal{T}}\longrightarrow (-\infty,0)$, $x\in \mathbb{R}$, $(k,\tau)\in\mathcal{K}\times\mathcal{T}$ and $(t,q)\in[0,T]\times\Z^{\#\mathcal{K}\times\#\mathcal{T}}$
\begin{align*}
v(t,q\ominus_{K_i,T_j} x):=v(t,q^{K_1,T_1},q^{K_1,T_2},\dots,q^{K_{i},T_{j+1}},q^{K_i,T_j}\!-\!x,q^{K_{i},T_{j+2}},\dots,q^{K_n,T_m}).
\end{align*}
The Hamilton-Jacobi-Bellman equation of the stochastic control problem \eqref{Reduced exchange problem} is given by
\begin{align}
\label{HJB Equation Principal First Form}
0=\partial_{t}v(t,q) +  \mathcal{H}_{E}\big(t, q,  v(t, \cdot)\big) , ~~ v(T, q) = -1,
\end{align}
with
$$
\mathcal{H}_{E}\big(t, q, v(t, \cdot)\big)  =  \sup_{z \in \mathcal{Z}} h_E\big( t, q, s, z, v(t, \cdot)\big),
$$
\begin{align*}
h_{E}\big( t, q, s, z , v(t, \cdot)\big) =& v(t,q) \Big(\frac{\eta}{2}\gamma\sigma^{2}\Big(\sum_{(k,\tau)\in \mathcal{K}\times\mathcal{T}}\Delta^{k,\tau}(z^{C^{k,\tau}}+q^{k,\tau})\Big)^{2} +\frac{\eta^{2}}{2}\sigma^{2}\Big(\sum_{(k,\tau)\in \mathcal{K}\times\mathcal{T}}\Delta^{k,\tau}z^{C^{k,\tau}}\Big)^{2}\Big) \\
&\hspace{1em}+\underset{i=a,b}{\sum} ~\sum_{(k,\tau)\in \mathcal{K}\times\mathcal{T}}  h^i_{k, \tau}\Big(t,  z^{k, \tau, i}, v(t,q ),v\big(t,q\ominus_{k,\tau}\phi(i) \big) \Big)\mathbf{1}_{\phi(i)\mathcal{Q}>-\overline{q}}
\end{align*}
and
\begin{align*}
h^i_{k, \tau}(t, z, y,y' ) &= \big( y'x_1^{k, \tau}e^{az} - yx_2e^{bz}\big) O_{k, \tau}
\end{align*}
where 
$$
x_1^{k, \tau} = e^{-\eta (c^{k, \tau}+\omega\big(\delta^{k, \tau}_{\infty}-\gamma^{-1}\log(1+\frac{\sigma\gamma}{C})\big)},~~ x_2 = \big(1 + \eta \frac{1 - (1+\frac{\sigma\gamma}{C})^{-1}}{\gamma} \big),~~ O_{k, \tau} = (1 + \frac{\sigma \gamma}{C})^{-\frac{C}{\gamma \sigma}}e^{-\frac{C}{\sigma}f^{k, \tau}},
$$
and
$$
a = \eta(1 - \omega) + \frac{C}{\sigma}, ~~b =  \frac{C}{\sigma}.
$$
Tedious but straightforward computations lead to the following optimizers:
\begin{align*}
\nonumber
& z^{\star k,\tau,i}:= \frac{1}{a-b}\log\big( \frac{bx_2v\big(t,q\big)}{ax_1^{k, \tau}v\big(t,q\ominus_{k,\tau}\phi(i) \big)}\big), \\
& z^{\star C^{k,\tau}}:=-\frac{\gamma}{\gamma+\eta}q^{k,\tau}.
\end{align*}
Note that from these computations, we get that this above optimization makes sense only if we assume that there exists $\delta_{\infty}$ large enough so that for $i= a$ or $b$, $k\in \mathcal{K},\tau\in \mathcal{T}$ and any $t,q$:
\begin{equation}
\label{eq:bound_condition}
\big| -z_{t}^{\star k,\tau,i}(t, q)+\frac{1}{\gamma}\textup{log}\Big(1+\frac{\sigma \gamma}{C}\Big)\big|<\delta_{\infty}.
\end{equation}
We will check that we can make such choice at the end of the verification argument. Equation \eqref{HJB Equation Principal First Form} is rewritten as
\begin{equation}
\label{PDE Principal}
0=\partial_{t}v(t,q) + v(t,q)  \frac{\gamma \eta^2}{\gamma+\eta}\frac{ \sigma^2 }{2}  \Big(\!\!\!\!\!\!\sum_{(k,\tau)\in \mathcal{K}\times\mathcal{T}}\!\!\!\!\!\!\Delta^{k,\tau}q^{k,\tau}\Big)^{2}\!\!-v(t,q)\underset{i=a,b}{\sum}\sum_{(k,\tau)\in \mathcal{K}\times\mathcal{T}}\!\!\!\!\!\tilde{C}^{k,\tau}\Big(\frac{v(t,q)}{v(t,q\ominus_{k,\tau}\phi(i))}\Big)^{\frac{C}{\sigma\eta(1-\omega)}}  \mathbf{1}_{\phi(i)\mathcal{Q}>-\overline{q}},
\end{equation}
where
\begin{align*}
\tilde{C}^{k, \tau} =& x_2(\frac{x_2}{x_1^{k, \tau}})^{ \frac{a}{a - b} } O_{k, \tau}\big( (\frac{b}{a})^{\frac{b}{a - b}}-  (\frac{b}{a})^{\frac{a}{ a-b}}  \big)>0.  
\end{align*}
We now make the ansatz $v(t, q) = u(t, \mathcal{Q})$. We derive the following PDE
\begin{equation}
\label{PDE Principal bis}
 0=\partial_{t}u(t,\mathcal{Q}) + u(t,\mathcal{Q})   \frac{\gamma \eta^2}{\gamma+\eta}\frac{ \sigma^2 }{2}   \mathcal{Q}^{2}-u(t,\mathcal{Q})\underset{i=a,b}{\sum}\sum_{(k,\tau)\in \mathcal{K}\times\mathcal{T}}\tilde{C}^{k,\tau}\Big(\frac{u(t,\mathcal{Q})}{u(t,\mathcal{Q}-\Delta^{k,\tau}\phi(i))}\Big)^{\frac{C}{\sigma\eta(1-\omega)}}\mathbf{1}_{\phi(i)\mathcal{Q}>-\overline{q}},
\end{equation}
with terminal condition $u(T, \mathcal{Q}) = -1$.\\

Using the classical change of variable $\tilde{u}:=(-u)^{-\frac{C}{\sigma\eta(1-\omega)}}$, PDE \eqref{PDE Principal bis} becomes

\begin{equation}
\label{PDE Principal Ter}
 0=\partial_{t}\tilde{u}(t,\mathcal{Q}) - \tilde{u}(t,\mathcal{Q})   \frac{C\gamma \eta}{\gamma+\eta}\frac{ \sigma }{2(1-\omega)}   \mathcal{Q}^{2}+\underset{i=a,b}{\sum}\sum_{(k,\tau)\in \mathcal{K}\times\mathcal{T}}\hat{C}^{k,\tau}\tilde{u}(t,\mathcal{Q}-\Delta^{k,\tau}\phi(i))\mathbf{1}_{\phi(i)\mathcal{Q}>-\overline{q}},
\end{equation}
where $\hat{C}^{k,\tau}:=\tilde{C}^{k,\tau}\frac{C}{\sigma\eta(1-\omega)}$. Eventually Cauchy-Lipschitz theorem provides existence and uniqueness of a bounded solution to \eqref{PDE Principal Ter} and so to \eqref{PDE Principal}.\\

For the verification argument, we first introduce a technical lemma. 

\begin{lemma} \label{Lemma Verification argument}
Let $Z\in \mathcal{Z},\xi=Y_{T}^{\hat{Y}_{0},Z}$. We define
\begin{align*}
K_{t}^{Z}:= \textup{exp}\bigg(-\eta\Big(\underset{i=a,b}{\sum}\sum_{(k,\tau)\in \mathcal{K}\times\mathcal{T}}c^{k,\tau}N_{t}^{k,\tau,i}-\int_{0}^{t}\omega\big(\Delta^{i}(Z_{s}^{k,\tau,i})-\delta^{k,\tau}_{\infty}\big)\mathrm{d}N_{s}^{k,\tau,i}-Y_{t}^{Y_{0},Z}\Big)\bigg),\quad t\in [0,T].
\end{align*} 
There exists $\epsilon>0$ such that
\begin{align*}
\underset{t\in [0,T]}{\textup{sup}}\mathbb{E}^{\hat\delta(Z)}\left[|K_{t}^{Z}|^{1+\epsilon}\right]<+\infty,
\end{align*}
where $\hat\delta(Z)$ is defined in Theorem \ref{Theorem Market Maker}.
\end{lemma}
The proof is borrowed from \cite{el2018optimal}. We now verify that the unique solution $v$ of Equation \eqref{HJB Equation Principal First Form} coincides at any point $(0, Q_0)$ with the value $v_0^E$ of the reduced problem \eqref{Reduced exchange problem}. We also prove that in \eqref{Reduced exchange problem}, the maximum is achieved for feedback controls issued from \eqref{exchange optimal control}. \\

Using It\^o's formula we get
\begin{align*}
\mathrm{d}[v(t, Q_t)K^{Z}_t] = K^Z_{t^-}\big[ & \Big(h_{E}\big(t, Q_{t^-}, S_t, Z_t, v(t, \cdot) \big) - \mathcal{H}_{E}\big(t, Q_{t^-}, v(t, \cdot)\big)\Big) \mathrm{d}t \\
&+  v(t, Q_{t^-})   \eta \sum_{k = 1}^{ \mathcal{K} } \sum_{\tau = 1}^{ \mathcal{T} } Z_t^{C^{k, \tau}} \mathrm{d}C_t^{k, \tau} \\
&  +  \sum_{i = a, b}  \sum_{k = 1}^{ \mathcal{K} } \sum_{\tau = 1}^{ \mathcal{T} } \big( e^{-\eta (c^{k, \tau} - Z^{k, \tau, i}_{t})} v( t, Q^{k, \tau}_{t^-} - \phi(i)) - v(t, Q_{t^-}) \big) \mathrm{d} N^{\hat{\delta}(Z), k, \tau, i}_t \big].
\end{align*}
The process $K^Z$ is uniformly integrable on $[0, T]$ according to \eqref{Strong integrability MM}, H\"older inequality and the boundedness of the intensity of the processes $N^{k, \tau, i}$. Moreover $v$ being uniformly bounded as a consequence of the Cauchy-Lipschitz theorem, the process $\big(v(t, Q_t)K^{Z}_t\big)_{t \in [0, T]}$ is a $\mathbb{P}^{\hat{\delta}(Z)}$ supermartingale and the local martingale term in the above equation is a true martingale. Hence
\begin{equation}
\label{eq:proof_verification_a}
v(0, Q_0)\geq \mathbb{E}^{\hat{\delta}(Z)}[v(T, Q_T)K^Z_T] = -\mathbb{E}^{\hat{\delta}(Z)}[K^Z_T].
\end{equation}
Since $Z\in \mathcal{Z}$ is arbitrary, we get
$$
v(0 ,Q_0) \geq \underset{Z \in \mathcal{Z} }{ \sup }-\mathbb{E}^{\hat{\delta}(Z)}[K^Z_T] = v_0^{E}.
$$
The feedback form of $Z$, issued from \eqref{exchange optimal control}, being bounded according to Equation \eqref{eq:bound_condition}, it is admissible. Considering the process $Z^\star$, we get an equality instead of an inequality in the above equation. \\ 

For consistency we now check that there does exist some constant $\delta_{\infty}$ such that \eqref{eq:bound_condition} is satisfied. In the same vein as in Step 2 of the proof of Theorem \ref{Theorem Market Maker}, we can show that for any $t$ and $q$, $v(t, q)$ is negative. Because of the compactness of the domain of $v$, the function is uniformly negative: we can find $\varepsilon$ such that $v < -\varepsilon$ on $[0, T]\times \mathcal{D}$. Consequently $ \log\Big( \frac{v(t, q)}{v\big(t, q^{k, \tau} \ominus_{k,\tau} \phi(i)\big)} \Big)$ is uniformly bounded in $i,~k, ~\tau,~t$ and $q$. Thus we can always choose a $\delta_{\infty}$ satisfying \eqref{eq:bound_condition}.
\endgroup

\bibliographystyle{plain}
\bibliography{biblio.bib}

\begin{thebibliography}{13}
\providecommand{\natexlab}[1]{#1}
\providecommand{\url}[1]{\texttt{#1}}
\expandafter\ifx\csname urlstyle\endcsname\relax
  \providecommand{\doi}[1]{doi: #1}\else
  \providecommand{\doi}{doi: \begingroup \urlstyle{rm}\Url}\fi

\bibitem[Baldacci et~al.(2019)Baldacci, Possama{\"\i}, and
  Rosenbaum]{baldacci2019optimal}
B.~Baldacci, D.~Possama{\"\i}, and M.~Rosenbaum.
\newblock Optimal make take fees in a multi market maker environment.
\newblock \emph{arXiv preprint arXiv:1907.11053}, 2019.

\bibitem[Budish et~al.(2015)Budish, Cramton, and Shim]{budish2015high}
E.~Budish, P.~Cramton, and J.~Shim.
\newblock The high-frequency trading arms race: Frequent batch auctions as a
  market design response.
\newblock \emph{The Quarterly Journal of Economics}, 130\penalty0 (4):\penalty0
  1547--1621, 2015.

\bibitem[Dayri and Rosenbaum(2016)]{dayri2015largetick}
K.~Dayri and M.~Rosenbaum.
\newblock Large tick assets: implicit spread and optimal tick size.
\newblock \emph{Market Microstructure and Liquidity}, 7\penalty0 (4), 2016.

\bibitem[El~Euch et~al.(2018)El~Euch, Mastrolia, Rosenbaum, and
  Touzi]{el2018optimal}
O.~El~Euch, T.~Mastrolia, M.~Rosenbaum, and N.~Touzi.
\newblock Optimal make--take fees for market making regulation.
\newblock \emph{arXiv preprint arXiv:1805.02741}, 2018.

\bibitem[Graf and Luschgy(2007)]{graf2007foundations}
S.~Graf and H.~Luschgy.
\newblock \emph{Foundations of quantization for probability distributions}.
\newblock Springer, 2007.

\bibitem[Jusselin et~al.(2019)Jusselin, Mastrolia, and
  Rosenbaum]{jusselin2019optimal}
P.~Jusselin, T.~Mastrolia, and M.~Rosenbaum.
\newblock Optimal auction duration: A price formation viewpoint.
\newblock \emph{Available at SSRN 3399151}, 2019.

\bibitem[Kalagnanam and Parkes(2004)]{kalagnanam2004auctions}
J.~Kalagnanam and D.~C. Parkes.
\newblock Auctions, bidding and exchange design.
\newblock In \emph{Handbook of Quantitative Supply Chain Analysis}, pages
  143--212. Springer, 2004.

\bibitem[Laruelle and Lehalle(2018)]{sophie2018market}
S.~Laruelle and C.-A. Lehalle.
\newblock \emph{Market microstructure in practice}.
\newblock World Scientific, 2018.

\bibitem[Madhavan et~al.(1997)Madhavan, Richardson, and
  Roomans]{madhavan1997security}
A.~Madhavan, M.~Richardson, and M.~Roomans.
\newblock Why do security prices change? a transaction-level analysis of nyse
  stocks.
\newblock \emph{The Review of Financial Studies}, 10\penalty0 (4):\penalty0
  1035--1064, 1997.

\bibitem[Mayhew and Mihov(2004)]{mayhew2004exchanges}
S.~Mayhew and V.~Mihov.
\newblock How do exchanges select stocks for option listing?
\newblock \emph{The Journal of Finance}, 59\penalty0 (1):\penalty0 447--471,
  2004.

\bibitem[Pag{\`e}s et~al.(2004)Pag{\`e}s, Pham, and Printems]{pages2004optimal}
G.~Pag{\`e}s, H.~Pham, and J.~Printems.
\newblock Optimal quantization methods and applications to numerical problems
  in finance.
\newblock In \emph{Handbook of computational and numerical methods in finance},
  pages 253--297. Springer, 2004.

\bibitem[Sokol(2013)]{sokol2013optimal}
A.~Sokol.
\newblock Optimal {N}ovikov-type criteria for local martingales with jumps.
\newblock \emph{Electronic Communications in Probability}, 18, 2013.

\bibitem[Wyart et~al.(2008)Wyart, Bouchaud, Kockelkoren, Potters, and
  Vettorazzo]{wyart2008relation}
M.~Wyart, J.-P. Bouchaud, J.~Kockelkoren, M.~Potters, and M.~Vettorazzo.
\newblock Relation between bid--ask spread, impact and volatility in
  order-driven markets.
\newblock \emph{Quantitative Finance}, 8\penalty0 (1):\penalty0 41--57, 2008.

\end{thebibliography}

\end{document}